\documentclass[11pt]{article}
\usepackage[latin1]{inputenc}
\usepackage[english]{babel}
\usepackage{amsmath}
\usepackage{amssymb}
\usepackage{amsthm}
\usepackage{anysize}
\usepackage[round]{natbib}
\usepackage{graphicx}
\usepackage{lscape}
\usepackage{times}
\usepackage{url}
\usepackage{multicol}
\usepackage{color}
\marginsize{35mm}{30mm}{20mm}{20mm}
\usepackage{subfig}
\usepackage{float}
\theoremstyle{plain}

\newtheorem{theorem}{Theorem}

\newtheorem{definition}[theorem]{Definition}

\newtheorem{lemma}[theorem]{Lemma}

\newtheorem{proposition}[theorem]{Proposition}

\newcommand{\tr}{\textup{Tr}}

\usepackage{color}

\newlength{\larg} % définit la taille des barres de la couverture
\begin{document}

\thispagestyle{empty}

\setcounter{page}{1}

\title{A flexible matrix Libor model with smiles}
\author{Jos\'{e} Da Fonseca\footnote{Auckland University of Technology, Department of Finance, Private Bag 92006, 1142 Auckland, New Zealand. email: jose.dafonseca@aut.ac.nz.}
  \and
\textrm{Alessandro Gnoatto}\thanks{Universit\`{a} degli Studi di Padova, Dipartimento di Matematica Pura ed Applicata, Via Trieste 63, Padova, Italy and Mathematisches Institut der LMU, Theresienstrasse 39, Munich  Germany. E-mail: gnoatto@mathematik.uni-muenchen.de}
\and
\textrm{Martino Grasselli}\thanks{Universit\`{a} degli Studi di Padova, Dipartimento di Matematica Pura ed Applicata, Via Trieste 63, Padova, Italy. E-mail: grassell@math.unipd.it and Ecole Sup\'{e}rieure d'Ing\'{e}nieurs L\'{e}onard de Vinci, D\'{e}partement Math\'{e}matiques et Ing\'{e}nierie Financi\`{e}re, 92916 Paris La D\'{e}fense, France.}}

\date{ \today}

\maketitle

\begin{abstract}
We present a flexible approach for the valuation of interest rate derivatives based on Affine Processes. We extend the methodology proposed in \cite{article_KRPT2009} by changing the choice of the state space. We provide semi-closed-form solutions for the pricing of caps and floors. We then show that it is possible to price swaptions in a multifactor setting with a good degree of analytical tractability. This is done via the Edgeworth expansion approach developed in \cite{article_CollinDufresneGoldstein2002}. A numerical exercise illustrates the flexibility of Wishart Libor model in describing the movements of the implied volatility surface.
\end{abstract}

\vspace{7cm}
\textbf{Keywords:} Affine processes, Wishart process, Libor market model, Fast Fourier Transform, Caps, Floors, Swaptions. \\

\textbf{JEL codes:} G13, C51.\\
\newpage

\section{Introduction}

In this paper we present a unified framework for the valuation of caps, floors and swaptions. These instruments are the most common derivative securities which are traded in a fixed income desk of a financial institution (see e.g. \cite{brigo06}). Practitioners usually price these products by relying on a Black-Scholes like formula, which was first presented in \cite{article_Black1976}. The market convention of pricing caps and swaptions using the Black formula is based on an application of the \cite{bla73} formula for stock options by assuming that the underlying interest rates are lognormally distributed. Remarkably, the use of this kind of formulae had no theoretical justification, since they involved a procedure in which the discount factor and the Libor rates were assumed to be independent so as to write the pricing formula as a product of a bond price and the expected payoff. The systematic use of this market practice ignited the interest of academics aiming at providing a coherent theoretical background.\\

In a series of articles, \cite{article_MiltersenSandmannSondermann1997}, \cite{article_BraceGatarekMusiela1997}, \cite{article_Jamshidian1997} and \cite{article_MusielaRutkowsky1997}, provided these theoretical foundations, introducing the Libor and Swap Market Model. Following these papers a stream of contributions appeared, trying to extend the basic model to the case where the volatility of the underlying factor is stochastic. The most famous proposals on this side can be found e.g. in \cite{article_AndersenBrotherton-Ratcliffe2001}, \cite{article_WuZhang2006},  \cite{article_JoshiRebonato2003}, \cite{andand2002}, \cite{pit05}, \cite{pit05b}. Other approaches explored different dynamics for the driving process with respect to the CEV or displaced-diffusion considered before for the Libor rate: for example \cite{glass03}, \cite{article_EberleinOzkan2005} introduced jump and more general Lévy processes, allowing for discontinuous sample paths of the driving process. Another interesting approach is the one of \cite{brigo03} based on a mixture of lognormals.\\
% A common pitfall consists in the problematic form of the variance-covariance matrix of the yields, which makes it difficult to perform a PCA analysis.\\

A typical problem in the previous approaches is that once the closed form solution for cap prices is found, to obtain an analogous result for swaptions it is customary to assume that the underlying (which is a coupon bond) behaves like a scalar process (typically again geometric Brownian motion). This results in inconsistencies between the so-called \textit{Libor} and \textit{Swap Market Models}. Even more important, by assuming that the coupon bond is driven by a scalar process, we do not take into account the correlation effects among the different coupons, a key feature of a swaption which may be viewed also as a correlation product. This last remark is of paramount importance for practitioners (see e.g. the introduction of \cite{article_CollinDufresneGoldstein2002}).\\

In this paper we consider a new approach based on the stochastic discount factor methodology, where instead of modeling directly the Libor rate, one concentrates on quotients of traded assets (i.e. bonds).  It has been first introduced by \cite{article_Constantinides92} and then developed by \cite{gou11} in a spot interest rate framework and by \cite{article_KRPT2009} in a Libor perspective. In this latter work they use affine processes on the state space $\mathbb{R}_{\geq0}^{d}$ as driving processes and provide a full characterization of the model, which allows them to provide closed form solutions for caps and swaptions up to Fourier integrals. This approach is very interesting and easily overcomes many difficulties which are to be faced in the computation of Radon-Nikodym derivatives.\\

We provide an extension of this approach, by considering affine processes on the state space $S_{d}^{++}$, the set of positive definite symmetric matrices. This state space may seem awkward at first sight, but the processes belonging to this family admit a characterization in terms of ODE's which resembles the one found for standard affine models, an example being given by the famous \cite{article_DuffieKan} model. In fact, in \cite{article_Cuchiero} the authors extend to the set $S_{d}^{+}$ (the set of positive semidefinite symmetric matrices) the classification of affine processes performed by \cite{article_DFS} for the state space $\mathbb{R}_{\geq0}^{d}\times \mathbb{R}^{n-d}$ introduced by \cite{article_DuffieKan}. What is more, the state space $S_{d}^{++}$ leads to stochastic factors which are non trivially correlated.  The most famous example of process defined in the set $S_{d}^{++}$ is the Wishart process, originally defined by \cite{article_Bru}, introduced in finance by \cite{article_gou02} and then extensively applied in \cite{article_gou03}, \cite{gou11}, \cite{article_DaFonseca1}, \cite{article_DaFonseca2}, \cite{article_DaFonseca3}, \cite{dafgra11}, among others.\\ 

The interesting feature of our framework is the possibility to obtain semi-closed form solutions for the pricing of swaptions in a multifactor setting, which is a well known challenging problem. In fact the exercise probability involves a multi-dimensional inequality. There have been many approaches to simplify the problem: for example, \cite{sinum2002} suggest an approximation of the exercise boundary
with a linear function of the state variables. However, the most efficient approach seems to be the one of \cite{article_CollinDufresneGoldstein2002} which heavily uses the affine structure of the model and is based on the Edgeworth expansion for the characteristic function in terms of the cumulants.  Since the cumulants decay very quickly the Edgeworth expansion for the exercise probability turns out to be very accurate and fast.\\

The paper is organized as follows. In Section 2 we introduce our framework by recalling some useful definitions and results on affine processes. Section 3 investigates the case of the state space $S_{d}^{++}$ and presents the technical results. In Section 4 we focus on the pricing problem of the relevant derivatives. Caps and Floors are briefly treated since their pricing is now quite standard within the FFT methodology, while we devote more attention to the pricing of swaptions by adopting the approach of \cite{article_CollinDufresneGoldstein2002}. Section 5 illustrates the flexibility of our framework through a numerical exercise. Section 6 concludes the paper, and we gather in the technical Appendices proofs and some remarks useful for implementation.\\

\section{Affine Processes on the set $S_{d}^{++}$ of strictly positive definite symmetric matrices}
\subsection{General results and notations}
To outline the setup we will consider affine processes taking values in the interior of the cone $S_d^+$. We will use the notations $\psi_t(u)=\phi(t,u)$ and $\phi_t(u)=\phi(t,u)$ so as to be consistent with \cite{article_KRPT2009}. We will be employing a property of the functions defining the Laplace transform, that we report after the following
 
%Affine processes have found many applications in Finance. For example, Gaussian interest rate models, like the ones of \cite{vas77} and \cite{hul93},  as well as square root models like in \cite{cox85}) belong to the class of affine models. \cite{heston}  developed a stochastic volatility model that can also be nested within the framewok of \cite{dar96}, who introduced in finance the notion of affine processes and allowed for both Gaussian and square root processes. In the last decade, there has been a number of theoretical and empirical applications in both interest rates and stochastic volatility modeling that systematically employed the affine structure for the stochastic factors.\\

%It is well known that it is possible to extend the notion of affine process, introduced by \cite{dar96} on the state space domain $\mathbb{R}_{\geq 0}^{m} \times \mathbb{R}^{n}$, to the domain $S_{d}^{+}$ (the set of symmetric positive semidefinite $d \times d$ matrices endowed with the scalar product $\left\langle x,y\right\rangle=Tr\left[ xy\right]$):
\begin{definition}
(\cite{article_Cuchiero}, Definition 2.1) Let $(\Omega, \mathcal{F} ,\left(\mathcal{F}_t \right)_{t \geq 0}, \mathbb{P})$ be a filtered probability space with the filtration $\left(\mathcal{F}_t \right)_{t \geq 0}$ satisfying the usual assumptions. A Markov process $\Sigma=\left(\Sigma_{t} \right)_{t \geq 0}$ with state space $S_{d}^{+}$, transition probability $p_{t}(\Sigma_0,A)=\mathbb{P}(\Sigma_{t} \in A)$ for $A \in S_{d}^{+}$, and transition semigroup $\left( P_{t}\right)_{t \geq 0}$ acting on bounded functions $f$ on $S_{d}^{+}$ is called affine process if:
\begin{enumerate}
\item it is stochastically continuous, that is, $\lim_{s \to t}p_{s}(\Sigma_0, \cdot)=p_{t}(\Sigma_0, \cdot)$ weakly on $S_{d}^{+}$ $\forall t \text{, } x \in S_{d}^{+}$, and
\item its Laplace transform has exponential-affine dependence on the initial state:
\begin{equation}
P_{t}e^{-\tr\left[u\Sigma_0 \right]}=\mathbb{E}\left[\left.e^{-\tr\left[u\Sigma_t\right]}\right|\mathcal{F}_0\right]=\int_{S_{d}^{+}}{e^{-\tr\left[u\xi \right]}p_t(\Sigma_0,d\xi)}=e^{-\phi_t(u)-\tr\left[ \psi_t(u)\Sigma_0\right]},
\label{LaplaceAffine}\end{equation}
$\forall t\text{ and } \Sigma_{0}, u \in S_{d}^{+}$, for some function $\phi: \mathbb{R}_{\geq 0}\times S_{d}^{+} \rightarrow \mathbb{R}_{\geq 0}$ and $\psi: \mathbb{R}_{\geq 0}\times S_{d}^{+} \rightarrow S_{d}^{+}$.
\end{enumerate}
\label{definition1}\end{definition}
%
%
%Note that in the definition above we assumed that the process is stochastically continuous, a feature that implies, according to Proposition 3.4 in \cite{article_Cuchiero}, that the process is regular in the following sense:
%\begin{definition} (\cite{article_Cuchiero}, Definition 2.2)
%The affine process $\Sigma$ is called regular if the derivatives
%\begin{equation}
%F(u)=\frac{\partial \phi(t,u)}{\partial t}|_{t=0^+}, \hspace{10mm}R(u)=\frac{\partial \psi(t,u)}{\partial t}|_{t=0^+}
%\end{equation}
%exist and are continuous at $u=0$.
%\label{definition2}\end{definition}
%
%
%

Having applications in mind, we will consider affine processes which are \textit{solvable} in the sense of \cite{article_Grasselli} (who investigated  affine processes on the more general symmetric cone state space domain): this means that the state space that we will consider is the interior of $S_{d}^{+}$, namely the cone of \textit{strictly positive definite} symmetric matrices, denoted by $S_{d}^{++}$\footnote{By analogy, the set of negative (resp. strictly negative) definite symmetric $d\times d$ matrices will be denoted by $S_{d}^{-}$ (resp. $S_{d}^{--}$).}.
 Solvability is important, in fact it ensures that the Riccati Ordinary Differential Equation associated to the Laplace transform (\ref{LaplaceAffine}) through the usual Feynman-Kac argument has a regular globally integrable flow: this will be crucial to outline our methodology (see e.g. the proof of Theorem \ref{TeoremaMartingale} in the sequel). \\

The next property closes our survey on affine processes. It will be needed when we prove that the structure of the model is preserved under changes of measure.
\begin{lemma}(\cite{article_Cuchiero} Lemma 3.2) Let $\Sigma$ be an affine process on $S_d^+$, then the functions $\phi$ and $\psi$ satisfy the following property:
\begin{align}
\phi_{t+s}(u)&=\phi_t(u)+\phi_s(\psi_t(u)),\nonumber\\
\psi_{t+s}(u)&=\psi_s(\psi_t(u)).\nonumber
\end{align}
\end{lemma}

\subsection{Examples}

The previous general framework may be quite abstract at a first sight, mostly because of the high technical level required to properly introduce the notion of admissibility and existence for affine processes, see \cite{article_Cuchiero}. In this subsection we provide some examples which will illustrate some concrete applications. We start with the most important one, which will also constitute our main object of study in the numerical illustrations.

\subsubsection{The Wishart process}\label{Wis_disc}

We suppose that the process $\Sigma$ is governed by the following (matrix) SDE:
\begin{equation}
d\Sigma_t=(\Omega\Omega^{\top}+M\Sigma_t+\Sigma_tM^{\top})dt+\sqrt{\Sigma_t}dW_tQ+Q^{\top}dW_t^{\top}\sqrt{\Sigma_t},\label{wishart}
\end{equation}
which was first studied by \cite{article_Bru} and whose solution is known as Wishart process. We assume $M,Q$ invertible and $M$ negative definite so as to ensure stationarity of the process. Moreover we require $\Omega\Omega^{\top}=\kappa Q^{\top}Q$ for a real parameter $\kappa \geq d+1$ to grant solvability (or equivalently to grant that $Det(\Sigma_t)>0$ with probability 1). Under the solvability assumption \cite{article_Grasselli} showed that the Riccati ODE corresponding to the characteristic function can be linearized and therefore admits a closed form solution. This is important in view of possible applications since in this case the functions $\phi$ and $\psi$ in definition (\ref{definition1}) are explicitly known:
\begin{proposition}
Consider the process $\Sigma=\left(\Sigma_t\right)_{0\leq t\leq T}$ which solves the SDE \eqref{wishart}. Then the conditional Laplace transform is given by:
\begin{align}
\mathbb{E}\left[\left.e^{-\tr\left[u\Sigma_T\right]}\right|\mathcal{F}_t\right]=e^{-\phi_\tau(u)-\tr\left[\psi_\tau(u)\Sigma_t\right]},
\end{align}
where $\tau:=T-t$. The functions $\phi_\tau(u)$ and $\psi_\tau(u)$ satisfy the following system of ODE's:
\begin{align}
\frac{\partial \psi_{\tau}}{\partial \tau}&=\psi_\tau(u) M+M^\top \psi_\tau(u)-2\psi_\tau(u) Q^\top Q\psi_\tau(u),\quad\psi_0(u)=u,\label{Riccati_psi}\\
\frac{\partial \phi_{\tau}}{\partial \tau}&=\tr\left[\kappa Q^{\top}Q\psi_\tau(u)\right],\quad\phi_0(u)=0\label{Riccati_phi}
\end{align}
which is solved by
\begin{align}
\psi_{\tau}(u)=\label{clever2}\left(u\psi_{12,\tau}(u)+\psi_{22,\tau}(u) \right)^{-1}\left(u\psi_{11,\tau}(u)+\psi_{21,\tau}(u)\right), 
\end{align}
where
\begin{equation}
\left(\begin{array}{rr}
\psi_{11,\tau}(u) & \psi_{12,\tau}(u)\\
\psi_{21,\tau}(u) & \psi_{22,\tau}(u)
\end{array}\right)
=\exp\left\{ \tau \left( 
\begin{array}{rr}
M & 2Q^{\top}Q\\
0 & -M^{\top}
\end{array}
\right) \right\}
\end{equation} 
and
\begin{equation}
\phi_{\tau}(u)=\frac{\kappa}{2}\tr\left[\log\left(u\psi_{12,\tau}(u)+\psi_{22,\tau}(u) \right)+M^{\top}\tau \right].
\end{equation}
\label{proplinearization}\end{proposition}
\begin{proof}
See  \cite{article_Grasselli}.
\end{proof}
The Wishart process constitutes the matrix analogue of the square root (Bessel) process. In fact we have that the matrix $M$ can be thought of as a mean reversion parameter: this is evident from the Lyapunov equation defining the long-run matrix $\Sigma_\infty$, which is given by
\begin{align}
-\kappa Q^{\top}Q=M\Sigma_\infty+\Sigma_\infty M^{\top}.\label{Lyap}
\end{align}

The second way to appreciate the analogies w.r.t the square root process is to look at the dynamics of the entries of the matrix process $\Sigma$. Concentrating on the main diagonal, in the $2\times 2$ case we have:
\begin{align}
d\Sigma_{11}&=\left(\kappa\left(Q_{11}^2+Q_{21}^2\right)+2\left(M_{11}\Sigma_{11}+M_{12}\Sigma_{12}\right)\right)dt\nonumber\\
&+2\sigma_t^{11}\left(Q_{11}dW_t^{11}+Q_{21}dW_t^{12}\right)+2\sigma_t^{12}\left(Q_{11}dW_{21}+Q_{21}dW_t^{22}\right)\\
d\Sigma_{22}&=\left(\kappa\left(Q_{22}^2+Q_{12}^2\right)+2\left(M_{21}\Sigma_{12}+M_{22}\Sigma_{22}\right)\right)dt\nonumber\\
&+2\sigma_t^{12}\left(Q_{12}dW_t^{11}+Q_{22}dW_t^{12}\right)+2\sigma_t^{22}\left(Q_{12}dW^{21}+Q_{22}dW^{22}\right)
\end{align}
where we set
\begin{align}
\left(\begin{array}{rr}
\sigma_{11}&\sigma_{12}\\
\sigma_{12}&\sigma_{22}
\end{array}\right):=\sqrt{\Sigma}.
\end{align}

We refer to \cite{article_DaFonseca4} for additional insights on the behavior of the Wishart process when aggregating its parameters.

\subsubsection{The pure jump OU process}\label{pure_jump}
The procedure we adopt in this paper is general, meaning that we can consider different examples of processes lying in the cone of positive definite matrices. In particular, we may consider the matrix subordinators proposed by \cite{article_BNS02}, or jump-diffusions like in \cite{article_LT}. In what follows we provide some examples with the calculations of the functions $\phi_{\tau}$ and $\psi_{\tau}$.\\

Let us consider the SDE
\begin{align}
d\Sigma_t=M\Sigma_t+\Sigma_tM^{\top}+dL_t,
\end{align}
where $M\in GL(d)$ is assumed as usual to be negative definite so as to grant stationarity, and $L_t$ is a pure jump process (compound Poisson Process) with constant intensity $\lambda$ and jump distribution $\nu$ with support on $S_d^{++}$.
The strong solution to this equation is given by:
\begin{align}
\Sigma_t=e^{Mt}\Sigma_0e^{M^{\top}t}+\int_{0}^{t}{e^{M(t-s)}dL_se^{M^{\top}(t-s)}}.
\end{align}
We are interested in the computation of the Laplace transform of this family of processes:
\begin{align}
\mathbb{E}\left[\left.e^{\tr\left[u\Sigma_T\right]}\right|\mathcal{F}_t\right]=e^{-\phi_\tau(u)-\tr\left[\psi_\tau(u)\Sigma_t\right]}.
\end{align}
The functions $\phi$ and $\psi$ solve the following (matrix) ODE's:
\begin{align}
\frac{\partial \psi_{{\tau}}}{\partial \tau}&=\psi_\tau(u)M+M^{\top}\psi_\tau(u)\quad\psi_0(u)=u\\
\frac{\partial \phi_{{\tau}}}{\partial \tau}&=-\lambda\int_{S_d^+\setminus\left\{0\right\}}{\left(e^{-\tr\left[\psi_\tau(u)\xi\right]}-1\right)\nu(d\xi)}\quad\phi_0(u)=0.
\end{align}
The solution for the first ODE is given by:
\begin{align}
\psi_\tau(u)=e^{M^{\top}\tau}ue^{M\tau},
\end{align}
so we can compute the Laplace transform by quadrature:
\begin{align}
\frac{\partial \phi_{\tau}}{\partial \tau}&=-\lambda\int_{S_d^+\setminus\left\{0\right\}}{\left(e^{-\tr\left[e^{M^{\top}\tau}ue^{M\tau}\xi\right]}-1\right)\nu(d\xi)}.
\end{align}
In the following we provide explicit computations by assuming some particular distribution $\nu(\cdot)$ for the jump size. The proofs of this formulae may be found in \cite{book_GuptaNagar2000}. For the sake of clarity, we specify that the Wishart distribution that we consider in the next sections are the classical distributions arising in the context of multivariate statistics.\\

\paragraph{\textbf{Wishart Distribution.}}
Let $J$ be the jump size. Consider the case $J\sim Wis_d\left(n,\mathcal{Q}\right)$. Then we have
\begin{align}
\phi_\tau(u)=-\lambda\int_{0}^{\tau}{\det\left(I_d+2e^{M^{\top}s}ue^{Ms}\mathcal{Q}\right)^{-\frac{n}{2}}ds}+\lambda \tau.
\end{align}
\paragraph{\textbf{Non-Central Wishart Distribution.}}
Let be $J\sim Wis_d\left(n,\mathcal{Q},\mathcal{M}\right)$, then we have
\begin{align}
\phi_\tau(u)&=-\lambda\int_{0}^{\tau}{\det\left(\mathcal{Q}\right)^{-\frac{n}{2}}\det\left(2e^{M^{\top}s}ue^{Ms}+\mathcal{Q}^{-1}\right)^{-\frac{n}{2}}}\times\nonumber\\
&\exp\left\{\tr\left[-\frac{1}{2}\mathcal{Q}^{-1}\mathcal{M}\mathcal{M}^{\top}+\frac{1}{2}\mathcal{Q}^{-1}\mathcal{M}\mathcal{M}^{\top}\mathcal{Q}^{-1}\left(2e^{M^{\top}s}ue^{Ms}+\mathcal{Q}^{-1}\right)\right]\right\}ds\nonumber\\
&+\lambda \tau.
\end{align}
\paragraph{\textbf{Beta type I distribution.}} Let be $J\sim\beta_d^I(a,b)$, then
\begin{align}
\phi_\tau(u)&=-\lambda\int_{0}^{\tau}{_{1}F_1(a;a+b;-e^{M^{\top}s}ue^{Ms})ds}+\lambda \tau.
\end{align}

\paragraph{\textbf{Beta  type II distribution.}} Let be $J\sim\beta_d^{II}(a,b)$, then
\begin{align}
\phi_\tau(u)&=-\lambda\int_{0}^{\tau}{\frac{\Gamma_d(a+b)}{\Gamma_d(b)}\Psi(a;-b+\frac{1}{2}(d+1);e^{M^{\top}s}ue^{Ms})ds}+\lambda \tau,
\end{align}
where $_{m}F_{n}$, $\Gamma_d(a)$, and $\Psi(a;b;R)$ denote respectively the hypergeometric function of matrix argument, the multivariate Gamma function and the confluent hypergeometric function, see e.g. \cite{book_GuptaNagar2000}.

\section{A Libor model on $S_{d}^{++}$ }
To outline the general framework for Libor models, we start by considering a filtered measurable space $\left(\Omega,\mathcal{F},\mathcal{F}_t\right)$ and a family of probability measures $\left(\mathbb{P}_{T_k}\right)_{1\leq k\leq N}$. Under the measure $\mathbb{P}_{T_N}$ we introduce a stochastic process $\Sigma$ taking values on the cone state space $S_d^{++}$. At this stage the process may be a diffusion, a pure jump or a jump-diffusion process taking values on $S_d^{++}$.
Consider a discrete tenor structure $0=T_{0} \leq T_{1} \leq...\leq T_{N} =T$.  We recall that the Libor rate is defined via quotients of bonds:
\begin{align}
L(t, T_{k}):=\frac{1}{\Delta T}\left(\frac{B(t,T_{k-1})}{B(t,T_{k})}-1\right),
\end{align}
where $\Delta T$ is assumed to be constant and $\Delta T=T_{k}-T_{k-1}$. The relation between the Libor rate and the forward price is given by:
\begin{align}
F(t,T_{k-1},T_k)=1+\Delta T L(t, T_{k}).
\end{align}
We proceed in full analogy with \cite{article_KRPT2009} by extending their results to processes taking values on the cone of positive definite matrices. The intuition is simple: to build up a Libor model with positive rates, quotients of bonds should be strictly greater than one. On the other hand, a no-arbitrage argument (see e.g. \cite{Geman_1995}) implies that quotients of bonds must be martingales under the forward risk neutral measure indexed by the maturity of the denominator, so that the key ingredient in the approach of \cite{article_KRPT2009} consists in the possibility of constructing a family of martingales that stay greater than one up to a bounded time horizon. This will be possible thanks to the affine structure of the model, since in this framework bond prices are exponentially affine in the positive (definite) factors, as well as their quotients.\\

\subsection{Martingales strictly greater than one}

Let us first define the set
\begin{center}
$\mathcal{I}_{T}:=\left\{ u\in S_{d}:\mathbb{E}\left[ e^{-\tr\left[u\Sigma_{T}\right]}\right]<\infty, \forall \Sigma_{0	} \in   S_{d}^{++}\right\}.$
\end{center}
 By the affine property of the process $\Sigma$ we have
\begin{align}
\mathbb{E}\left[ e^{-\tr\left[u\Sigma_{t}\right]}\right]&=e^{-\phi_t(u)-\tr\left[\psi_t(u)\Sigma_0\right]},\nonumber\\
\phi:&\left[0,T\right]\times\mathcal{I}_{T}\rightarrow \mathbb{R},\nonumber\\
\psi:&\left[0,T\right]\times\mathcal{I}_{T}\rightarrow S_d.
\end{align}

Within this setting we are able to construct martingales that stay greater than one up to a bounded time horizon $T$.
\begin{theorem}
Let $\Sigma$ be an affine process, and let $u\in\mathcal{I}_{T}\cap S_{d}^{--}$, then the process $M^{u}$ defined by
\begin{equation}
M^{u}_{t}=\exp\left\{ -\phi_{T-t}(u)-\tr \left[ \psi_{T-t}(u)\Sigma_{t} \right] \right\}\label{martingale}
\end{equation}
is a martingale and $M^{u}_{t}>1$ a.s. $\forall t \in \left[0,T \right].$
\label{TeoremaMartingale}\end{theorem}

\begin{proof} See Appendix. \end{proof}

Equipped with this positivity result, we can proceed by considering a tenor structure of non negative Libor rates $L(0, T_{k})$ for $k=\left\{ 1,...,N-1\right\}$. Standard arbitrage arguments (see e.g. \cite{Geman_1995}) imply that discounted traded assets, in our case bonds, are martingales under the terminal martingale measure:
\begin{equation}
\frac{B(*,T_{k})}{B(*,T_{N})}\in \mathcal{M}\left( \mathbb{P}_{T_{N}}\right)\quad\forall k \in \left\{ 1,...,N-1\right\},
\end{equation}
where $\mathcal{M}\left( \mathbb{P}_{T_{N}}\right)$ denotes the set of martingales with respect to the forward risk neutral probability $\mathbb{P}_{T_{N}}$. The idea in \cite{article_KRPT2009} is then to model quotients of bond prices using the martingales $M^{u}$ defined as follows:
\begin{align}
\frac{B(t,T_{1})}{B(t,T_{N})}&=M^{u_{1}}_{t}\label{bondstructure1}\\
&\vdots\nonumber\\
\frac{B(t,T_{N-1})}{B(t,T_{N})}&=M^{u_{N-1}}_{t}	\label{bondstructureN}
\end{align}
$\forall t \in \left[ 0, T_{1}\right],...,t \in \left[0,T_{N-1} \right]$ respectively. As a consequence, the initial values of the martingales $M_{0}^{u_{k}}$ must satisfy the relation
\begin{equation}
M_{0}^{u_{k}}=\exp\left\{ -\phi_{T}(u_{k})-\tr\left[\psi_{T}(u_{k})\Sigma_{0} \right]\right\}=\frac{B(0,T_{k})}{B(0,T_{N})},
\label{initialbondstructure}\end{equation}
for all $k \in \left\{ 1,...,N-1\right\}$, so that it is possible to set $u_{N}=0$ as we have $M_{0}^{u_{N}}=1$.\\

In the following proposition, we show that it is possible to fit (basically) any initial term structure of bond rates. The state space we are considering offers a wide range of possibilities to perform this task. However, since we are interested in applications, we adopt the simplest choice directly coming from the scalar case and we focus on the particular (but realistic) case where all Libor rates are positive.

\begin{proposition}
Let $L(0, T_{1}),...,L(0, T_{N})$ be a tenor structure of positive initial Libor rates, and let $\Sigma$ be an affine process on $S_{d}^{++}$. Define
\begin{equation}
\gamma_{\Sigma}:=\sup_{u \in \mathcal{I}_{T} \cap S_{d}^{--}}\mathbb{E}\left[ e^{-\tr\left[u\Sigma_{T}\right]}\right].
\label{gammasigma}\end{equation}
%\begin{enumerate}
If $\gamma_{\Sigma}>\frac{B(0,T_{1})}{B(0,T_{N})}$ then there exists a strictly increasing sequence of matrices (i.e. $u_{k} \prec u_{k+1}$ if and only if $u_{k}-u_{k+1} \in S_{d}^{--}$) $u_{1} \prec u_{2} \prec ...\prec u_{N-1} \prec 0$ in $\mathcal{I}_{T} \cap S_{d}^{--}$ and $u_N=0$ such that
\begin{equation}
M_{0}^{u_{k}}=\frac{B(0,T_{k})}{B(0,T_{N})},\hspace{10mm}\forall k \in \left\{ 1,...,N\right\}.
\end{equation}
Conversely, let the bond prices be given by (\ref{bondstructure1})-(\ref{bondstructureN})  and satisfy the initial condition (\ref{initialbondstructure}). Then the Libor rates $L(t,T_{k})$ are positive a.s. $\forall t \in \left[0, T_{k} \right]$ and $k\in \left\{ 1,...,N-1\right\}.$
%In particular, if $\gamma_{\Sigma}=\infty$ then the affine Libor model can fit any term structure of non-negative initial Libor rates.
%\item If all Libor rates are positive, then the sequence $\left( u_{k} \right)_{k \in \left\{ 1,...,N\right\}}$ is strictly decreasing.
%\end{enumerate}
\label{propgamma}\end{proposition}
\begin{proof} See Appendix. \end{proof}

\subsection{A fully-affine arbitrage-free model}

If we look at the definition of the Libor rate we realize that it is quite natural to require quotients of bonds to be driven by an exponentially affine function of the state: in fact, in this case also bond prices as well as forward prices will be affine functions. This is also in line with the previous approaches of \cite{article_Constantinides92} and \cite{gou11} based on the stochastic discount factor. In other words, the approach of \cite{article_KRPT2009} is able to provide a fully affine structure\footnote{This is the reason why we will be able to apply the approach  by \cite{article_CollinDufresneGoldstein2002}, who originally started from an affine short rate in order to price swaptions: in fact, also in their framework bond prices are affine functions of the state variables.}. To prove the affine structure or our model, we first show that under (\ref{bondstructure1})-(\ref{bondstructureN}), forward prices are of exponential-affine form
under any forward measure. To do this, first we notice that in this framework quotients of bonds are exponentially affine in the state factors, so that also forward prices will be: for $k=1,...,N-1$
\begin{align}
\frac{B(t,T_{k})}{B(t,T_{k+1})}&=\frac{B(t,T_{k})}{B(t,T_{N})}\frac{B(t,T_{N})}{B(t,T_{k+1})}=\frac{M_{t}^{u_{k}}}{M_{t}^{u_{k+1}}}\nonumber\\
&=\exp\left\{ -\phi_{T_{N}-t}(u_{k})+\phi_{T_{N}-t}(u_{k+1})\right\} \nonumber\\
&\exp\left\{\tr\left[\left(-\psi_{T_{N}-t}(u_{k})+\psi_{T_{N}-t}(u_{k+1})\right)\Sigma_{t}\right]\right\}\nonumber\\
&=:\exp\left\{A_{T_{N}-t}(u_{k},u_{k+1})+\tr\left[B_{T_{N}-t}(u_{k},u_{k+1})\Sigma_{t}\right]\right\}.\label{exponentiallyAB}
\end{align}
With this result, we are able to show very easily that the model is arbitrage free, that is forward prices are martingales with respect to their corresponding forward measures (see \cite{Geman_1995}):
\begin{equation}
\frac{B(*,T_{k})}{B(*,T_{N})}\in \mathcal{M}\left( \mathbb{P}_{T_{N}}\right).
\end{equation}
This comes from the fact that forward measures are related one another via the quotients of the martingales $M^{u}$:
\begin{equation}
\frac{\partial \mathbb{P}_{T_{k}}}{\partial \mathbb{P}_{T_{k+1}}} \arrowvert_{\mathcal{F}_{t}}=\frac{F(t,T_{k},T_{k+1})}{F(0,T_{k},T_{k+1})}=\frac{B(0,T_{k+1})}{B(0,T_{k})}\frac{M^{u_{k}}_{t}}{M^{u_{k+1}}_{t}},
\end{equation}
$\forall k \in \left\{ 1,...,N\right\}$. Then $L(*,T_{k})$ is a martingale under the forward measure $\mathbb{P}_{T_{k+1}}$ since the successive densities from $\mathbb{P}_{T_{k+1}}$ to $\mathbb{P}_{T_{N}}$ yield a telescoping product and a $\mathbb{P}_{T_{N}}$ martingale (see \cite{article_KRPT2009}). More precisely:
\begin{equation}
1+\Delta T L(*,T_{k})=\frac{B(*,T_{k})}{B(*,T_{k+1})}=\frac{M^{u_{k}}}{M^{u_{k+1}}}\in \mathcal{M}\left( \mathbb{P}_{T_{k+1}}\right)
\end{equation}
since
\begin{equation}
\frac{M^{u_{k}}}{M^{u_{k+1}}}\prod_{l=k+1}^{N-1}\frac{M^{u_{l}}}{M^{u_{l+1}}}= M^{u_{k}} \in \mathcal{M} \left( \mathbb{P}_{T_{N}}\right).
\end{equation}
Also, the density between the $\mathbb{P}_{T_{k}}$-forward measure and the terminal forward measure $\mathbb{P}_{T_{N}}$ is given by the martingale $M^{u_{k}}$ as indicated by (\ref{bondstructure1})-(\ref{bondstructureN}):
\begin{equation}
\frac{\partial \mathbb{P}_{T_{k}}}{\partial \mathbb{P}_{T_{N}}} \arrowvert_{\mathcal{F}_{t}}=\frac{B(0,T_{N})}{B(0,T_{k})}\frac{B(t,T_{k})}{B(t,T_{N})}=\frac{B(0,T_{N})}{B(0,T_{k})}M^{u_{k}}_{t}=\frac{M^{u_{k}}_{t}}{M^{u_{k}}_{0}}.\label{change_measure}
\end{equation}

In this arbitrage-free model with positive Libor rates, the affine structure is preserved: that is, it is possible to extend to the state space $S_d^{++}$ the analogous result of \cite{article_KRPT2009}.

\begin{proposition}Let the bond structure be defined through (\ref{bondstructure1})-(\ref{bondstructureN}), where the process $M^{u_.}$ is given by (\ref{martingale}). Then forward prices are exponentially affine in the state variable $\Sigma$ under any forward measure.
\end{proposition}

\begin{proof} The result comes directly from formula (6.23) in \cite{article_KRPT2009} once the scalar product is replaced  by the trace operator. \end{proof}

\section{Pricing of Derivatives}

We now focus on the pricing problem for vanilla options like Caps, Floors and for exotic options like swaptions in the affine Libor model on $S^{++}_d$ introduced in the previous section. We shall see that the pricing of Caps and Floors may be performed using standard Fourier pricing techniques as in \cite{article_KRPT2009}, whereas, for the case of swaptions, we will resort to a quasi closed form solution. In fact, since the moments of the underlying affine process are known through its characteristic function, we can expand the exercise probability via an Edgeworth development, as shown in \cite{article_CollinDufresneGoldstein2002}. This approach will lead to an efficient approximation that will avoid the numerical problems underlying the computation of the exercise probability in  \cite{article_KRPT2009}.

\subsection{Caps and Floors}
A Cap may be thought of as a portfolio of call options on the successive Libor rates, named Caplets, whereas Floors are portfolios of put options named floorlets. These options are usually settled \textit{in arrears}, which means that the caplet with maturity $T_{k}$ is settled at time $T_{k+1}$. The tenor length $\Delta T$ is assumed to be constant. Since the two products are equivalent, we will focus on Caps. A Cap with unit notional has a payoff given by the following:
\begin{equation}
\Delta T \left( L(T_{k},T_{k})-K\right)^{+}\hspace{10mm}k=1,...,N-1
\end{equation}
We rewrite the payoff of caplets as in \cite{article_KRPT2009}:
\begin{align}
\Delta T \left( L(T_{k},T_{k})-K\right)^{+}&= \left(1+ \Delta TL(T_{k},T_{k})-\left(1+\Delta TK\right)\right)^{+}\nonumber\\
&=\left(\frac{M^{u_{k}}_{T_{k}}}{M^{u_{k+1}}_{T_{k}}}-\mathcal{K} \right)^{+},
\end{align}
with $\mathcal{K}:=1+\Delta TK$.

Thus we see that the caplet is equivalent to an option on the forward price. In order to avoid the computation of expectations involving a joint distribution, each single caplet is priced under the corresponding forward measure:
\begin{align}
\mathbb{C}\left(T_{k},K\right)&=B(0,T_{k+1})\mathbb{E}^{\mathbb{P}_{T_{k+1}}}\left[\left( \frac{M^{u_{k}}_{T_{k}}}{M^{u_{k+1}}_{T_{k}}}-\mathcal{K} \right)^{+} \right]\nonumber\\
&=B(0,T_{k+1})\mathbb{E}^{\mathbb{P}_{T_{k+1}}}\left[\left(e^Y-\mathcal{K} \right)^{+} \right],
\end{align}
with:
\begin{align}
Y:=\log\left(\frac{M^{u_{k}}_{T_{k}}}{M^{u_{k+1}}_{T_{k}}}\right)=A_{T_N-T_k}(u_k,u_{k+1})+\tr\left[B_{T_N-T_k}(u_k,u_{k+1})\Sigma_{T_{k}} \right],
\end{align}
for $A_{T_N-T_k}(u_k,u_{k+1})$, $B_{T_N-T_k}(u_k,u_{k+1})$ defined as in \eqref{exponentiallyAB}. The pricing problem can be solved via Fourier techniques through the \cite{article_Carr99} methodology. Hence we have the following proposition, whose standard proof is omitted. %We skip the details of the simple proof and we refer to \cite{article_KRPT2009} for the derivation of the %following 
\begin{proposition}\label{price_caplet}
Let $\alpha>0$. The price of a caplet with strike K and maturity $T_{k}$ is given by the formula:
\begin{align}
\mathbb{C}\left(T_{k},K\right)&=B(0,T_{k+1})\frac{\exp\left\{-\alpha c \right\}}{2\pi}\nonumber\\
&\times\int_{-\infty}^{+\infty}{e^{-ivc} \frac{\mathbb{E}^{\mathbb{P}_{T_{k+1}}}\left[e^{i\left( v-\left(\alpha +1 \right)i\right)\left(A_{T_N-T_k}(u_k,u_{k+1})+\tr\left[B_{T_N-T_k}(u_k,u_{k+1})\Sigma_{T_{k}} \right]\right)} \right]}{\left(\alpha +iv\right)\left(1+\alpha +iv\right)}dv},
\end{align}
where:
\begin{align}
c&=\log\left(1+\Delta TK\right),\nonumber\\
A_{T_N-T_k}(u_k,u_{k+1})&=-\phi_{T_{N}-T_{k}}(u_{k})+\phi_{T_{N}-T_{k}}(u_{k+1}),\nonumber\\
B_{T_N-T_k}(u_k,u_{k+1})&=-\psi_{T_{N}-T_{k}}(u_{k})+\psi_{T_{N}-T_{k}}(u_{k+1}).\nonumber
\end{align}
\end{proposition}

In other words, pricing a Cap involves the computation of the moment generating function of e.g. the Wishart process, which can be efficiently performed through the linearization of the associated Riccati ODEs as explained in Proposition \ref{proplinearization}. The parameter $\alpha>0$ represents the damping factor introduced by \cite{article_Carr99}. We report in the Appendix B the explicit expression of the characteristic function involved in the pricing procedure.

\subsection{Swaptions}
The payoff of a receiver (resp. payer) swaption may be seen as a call (resp. put) on a coupon bond with strike price equal to one. We consider a receiver swaption that starts at $T_{i}$ with maturity $T_{m}$, $(i<m\leq N)$. The time-$T_{i}$ value is given by:
\begin{equation}
\mathbb{S}_{T_{i}}(K,T_{i},T_{m})=\left(\sum_{k=i+1}^{m}{c_{k}B(T_{i},T_{k})}-1 \right)^{+}
\end{equation}
where
\begin{equation}
c_{k} = \left\{
\begin{array}{ll}
\Delta TK & \text{if } i+1\leq k \leq m-1,\\
1+\Delta TK & \text{if } k=m.
\end{array} \right.
\end{equation}
Unfortunately, we face some difficulties if we try to adopt the Fourier technique that we employed to price a caplet. To see this we look at the proof of Proposition 7.2. in \cite{article_KRPT2009}, which requires the computation of the Fourier transform of the payoff\footnote{$B(T_{i},T_{k})=\frac{B(T_{i},T_{k})}{B(T_{i},T_{N})}\frac{B(T_{i},T_{N})}{B(T_{i},T_{i})}=\frac{M^{u_k}_{T_i}}{M^{u_i}_{T_i}}=\exp\left\{A_{T_N-T_i}(u_k,u_{i})+\tr\left[B_{T_N-T_i}(u_k,u_{i})\Sigma_{T_i}\right]\right\}$}:
\begin{align}
\tilde{f}(z)&=\int_{\mathbb{R}^{\frac{d(d+1)}{2}}}{e^{\tr\left[iz\Sigma_{T_i}\right]}\left(\sum_{k=i+1}^{m}{c_{k}e^{A_{T_N-T_i}(u_k,u_{i})+\tr\left[B_{T_N-T_i}(u_k,u_{i})\Sigma_{T_i} \right]}} -1\right)^{+}dvech(\Sigma_{T_i})},
\end{align}
where for a symmetric matrix $A$, $vech(A)$ stands for the vector in $\mathbb{R}^{d(d+1)/2}$ consisting in the columns of the upper-diagonal part of $A$
including the diagonal.
The problem is given by the presence of the positive part in the payoff function. To get rid of it, we should be able to find a value $\tilde{\Sigma}$ such that
\begin{equation}
\sum_{k=i+1}^{m}{c_{k}e^{A_{T_N-T_i}(u_k,u_{i})+\tr\left[B_{T_N-T_i}(u_k,u_{i})\tilde{\Sigma} \right]}}=1, \label{findBound}
\end{equation}
that is we should solve a single equation in $d(d+1)/2$ unknowns (the elements of $\tilde{\Sigma}$), which is highly non trivial when $d>1$.  Thus, pricing swaptions is challenging when we consider multiple factor affine models: this is a well known problem, see e.g. \cite{article_Jam89} and  \cite{article_CollinDufresneGoldstein2002}. \cite{article_KRPT2009} investigate the case $d=1$, that is a Libor model driven by a (univariate) CIR process like in \cite{article_Jam87}. In that case, solving an  equation similar to (\ref{findBound}) is simple and the pricing of a swaption is only slightly more numerically complicated than the pricing of a Cap. As our purpose is to extend their methodology to a process with values in the set of strictly positive definite symmetric matrices we face a numerical difficulty related to the dimension of the state space. In order to solve this difficulty we follow   \cite{article_CollinDufresneGoldstein2002}'s methodology which strongly depends on the affine property of the process used to modelize the rates. As the processes we use have this affine property we can carry out the approximation for the swaption price proposed by these authors. Therefore, we can get around the dimensional difficulties posed by the process.\\

We briefly recall the main results of  \cite{article_CollinDufresneGoldstein2002} to approximate the exercise probabilities for the swaption.
We define the $T_{i}$-price of a coupon bond, for $i<m\leq N$, as follows:
\begin{equation}
CB(T_{i})=\sum_{k=i+1}^{m}{c_{k}B(T_{i},T_{k})}.
\end{equation}

 Let us derive the general form of the pricing formula for a receiver swaption, for $0=T_{0}=t <T_{i}$:
\begin{align}
\mathbb{S}_{0}(K,T_{i},T_{m})&=\mathbb{E}^{\mathbb{Q}}\left[e^{-\int_{0}^{T_{i}}{r_s ds}}\left(CB(T_{i})-1 \right)^{+} \right]\nonumber\\
&=\mathbb{E}^{\mathbb{Q}}\left[e^{-\int_{0}^{T_{i}}{r_s ds}}\left(CB(T_{i})\textbf{1}_{\left(CB(T_{i})>1\right)}-\textbf{1}_{\left(CB(T_{i})>1\right)} \right) \right]\nonumber\\
%&=\sum_{k=i+1}^{m}{c_{k}\mathbb{E}^{\mathbb{Q}}\left[e^{-\int_{0}^{T_{i}}{r_s %ds}}\textbf{1}_{\left(CB(T_{i})>1\right)}B(T_{i},T_{k})\right]}\nonumber\\
%&-\mathbb{E}^{\mathbb{Q}}\left[e^{-\int_{0}^{T_{i}}{r_s ds}}\textbf{1}_{\left(CB(T_{i})>1\right)}\right]\nonumber\\
%&=\sum_{k=i+1}^{m}{c_{k}\mathbb{E}^{\mathbb{Q}}\left[e^{-\int_{0}^{T_{i}}{r_s %ds}}\textbf{1}_{\left(CB(T_{i})>1\right)}\mathbb{E}^{\mathbb{Q}}\left[e^{-\int_{T_{i}}^{T_{k}}{r_s ds}}|\mathcal{F}_{T_{i}}\right]\right]}\nonumber\\
%&-K\mathbb{E}^{\mathbb{Q}}\left[e^{-\int_{0}^{T_{i}}{r_s ds}}\textbf{1}_{\left(CB(T_{i})>1\right)}\right]\nonumber\\
&=\sum_{k=i+1}^{m}{c_{k}\mathbb{E}^{\mathbb{Q}}\left[e^{-\int_{0}^{T_{k}}{r_s ds}}\textbf{1}_{\left(CB(T_{i})>1\right)}\right]}\nonumber\\
&-\mathbb{E}^{\mathbb{Q}}\left[e^{-\int_{0}^{T_{i}}{r_s ds}}\textbf{1}_{\left(CB(T_{i})>1\right)}\right].\nonumber
\end{align}
We switch to the forward measure as follows:
\begin{align}
\mathbb{S}_{0}(K,T_{i},T_{N})&=\sum_{k=i+1}^{m}{c_{k}B(0,T_{k})\mathbb{E}^{\mathbb{Q}}\left[\frac{e^{-\int_{0}^{T_{k}}{r_s ds}}}{B(0,T_{k})}\textbf{1}_{\left(CB(T_{i})>1\right)}\right]}\nonumber\\
&-B(0,T_{i})\mathbb{E}^{\mathbb{Q}}\left[\frac{e^{-\int_{0}^{T_{i}}{r_s ds}}}{B(0,T_{i})}\textbf{1}_{\left(CB(T_{i})>1\right)}\right]\nonumber\\
&=\sum_{k=i+1}^{m}{c_{k}B(0,T_{k})\mathbb{E}^{\mathbb{P}_{T_{k}}}\left[\textbf{1}_{\left(CB(T_{i})>1\right)}\right]}\nonumber\\
&-B(0,T_{i})\mathbb{E}^{\mathbb{P}_{T_{i}}}\left[\textbf{1}_{\left(CB(T_{i})>1\right)}\right]\nonumber\\
%&=\sum_{k=i+1}^{N}{C_{k}B(t,T_{k})\mathbb{E}^{\mathbb{P}_{T_{k}}}\left[\textbf{1}_{\left(CB(T_{i})>K\right)}|\mathcal{F}_{t}\right]}\nonumber\\
%&-KB(t,T_{i})\mathbb{E}^{\mathbb{P}_{T_{i}}}\left[\textbf{1}_{\left(CB(T_{i})>K\right)}|\mathcal{F}_{t}\right]\nonumber\\
&=\sum_{k=i+1}^{m}{c_{k}B(0,T_{k})\mathbb{P}_{T_{k}}\left[\left(CB(T_{i})>1\right)\right]}\nonumber\\
&-B(0,T_{i})\mathbb{P}_{T_{i}}\left[\left(CB(T_{i})>1\right)\right].\nonumber
\end{align}

The exercise probabilities $\mathbb{P}_{T_{k}}\left[\left(CB(T_{i})>1\right)\right]$ and $\mathbb{P}_{T_{i}}\left[\left(CB(T_{i})>1\right)\right]$ do not admit in general a closed form expression but can be efficiently approximated thanks to an Edgeworth expansion procedure. Intuitively, the moments of the coupon bonds admit a simple closed-form expression in our affine framework, and these moments uniquely identify the cumulants of the distribution. One can expand the characteristic function in terms of the cumulants and compute the exercise probabilities by Fourier inversion.\\

Using the notation of \cite{article_CollinDufresneGoldstein2002} (their formula (5)) for the $q-th$ power of a coupon bond we notice that, for $i<m\leq N$:
\begin{align}
\left( CB(T_{i})\right)^{q}&=\left(c_{i+1}B(T_{i},T_{i+1})+...+c_{m}B(T_{i},T_{m}) \right)^{q}\nonumber\\
&=\sum_{j_{1},...,j_{q}=i+1}^{m}{\left(c_{j_{1}}\cdot...\cdot c_{j_{q}} \right)\times\left(B(T_{i},T_{j_{1}})\cdot...\cdot B(T_{i},T_{j_{q}}) \right)}.
\end{align}

Now in our framework we have (see also formula (7.9) in \cite{article_KRPT2009})
\begin{equation}
B(T_{i},T_{j_l})=\frac{M^{u_{j_l}}_{T_{i}}}{M^{u_{i}}_{T_{i}}}
\end{equation}
for $l=1,...,q$, meaning that we can rewrite the $q-th$ power of the coupon-bond as follows:
\begin{equation}
\left( CB(T_{i})\right)^{q}=\sum_{j_{1},...,j_{q}=i+1}^{m}{\left(c_{j_{1}}\cdot...\cdot c_{j_{q}} \right)\times\left(\frac{M^{u_{j_{1}}}_{T_{i}}}{M^{u_i}_{T_{i}}}\cdot...\cdot \frac{M^{u_{j_{q}}}_{T_{i}}}{M^{u_i}_{T_{i}}} \right)}.
\end{equation}
Recall, from \eqref{martingale}, that we have
\begin{equation}
M^{u_{j_{l}}}_{T_{i}}=\exp \left\{-\phi_{T_{N}-T_{i}}(u_{j_{l}})-\tr\left[ \psi_{T_{N}-T_{i}}(u_{j_{l}})\Sigma_{T_{i}}\right] \right\},
\end{equation}
for $l=1,...,q$ and
\begin{equation}
M^{u_{i}}_{T_{i}}=\exp \left\{-\phi_{T_{N}-T_{i}}(u_{i})-\tr\left[ \psi_{T_{N}-T_{i}}(u_{i})\Sigma_{T_{i}}\right] \right\}.
\end{equation}

In conclusion, the $q-th$ moment under $\mathbb{P}_{T_k}$ has the following expression:
\begin{align}
&\mathbb{E}^{\mathbb{P}_{T_{k}}}\left[CB(T_{i})^q \right]\nonumber\\
&=\sum_{j_{1},...,j_{q}=i+1}^{m}\left(c_{j_{1}}\cdot...\cdot c_{j_{q}} \right)\times{\mathbb{E}^{\mathbb{P}_{T_{k}}}\left[\left(\frac{M^{u_{j_{1}}}_{T_{i}}}{M^{u_i}_{T_{i}}}\cdot...\cdot \frac{M^{u_{j_{q}}}_{T_{i}}}{M^{u_i}_{T_{i}}} \right) \right]}\nonumber\\
&=\sum_{j_{1},...,j_{q}=i+1}^{m}\left(c_{j_{1}}\cdot...\cdot c_{j_{q}} \right)\times\nonumber\\
&\mathbb{E}^{\mathbb{P}_{T_{k}}}\Bigg[\exp\Bigg\{ \sum_{l=1}^{q}{\Big(-\phi_{T_{N}-T_{i}}(u_{j_{l}})-\tr\left[ \psi_{T_{N}-T_{i}}(u_{j_{l}})\Sigma_{T_{i}}\right]\Big)}\Bigg. \nonumber\\
&\Bigg.+q\Big( \phi_{T_{N}-T_{i}}(u_{i})+\tr\left[ \psi_{T_{N}-T_{i}}(u_{i})\Sigma_{T_{i}}\right]\Big)\Bigg\}\Bigg]\nonumber\\
&=\sum_{j_{1},...,j_{q}=i+1}^{m}\left(c_{j_{1}}\cdot...\cdot c_{j_{q}} \right)\times \exp \left\{\left(-\sum_{l=1}^{q}{\phi_{T_{N}-T_{i}}(u_{j_{l}})}\right)+q\phi_{T_{N}-T_{i}}(u_{i}) \right\}\nonumber\\
&\times \mathbb{E}^{\mathbb{P}_{T_{k}}}\left[\exp\left\{\tr\left[\left(\left(-\sum_{l=1}^{q}{\psi_{T_{N}-T_{i}}(u_{j_{l}})}\right)+q\psi_{T_{N}-T_{i}}(u_{i})\right)\Sigma_{T_{i}} \right] \right\} \right],
\end{align}
where the functions $\phi$ and $\psi$ are as usual the solutions of Riccati ODE's of the form \eqref{Riccati_psi}, \eqref{Riccati_phi}. Once the first $m$ moments under the corresponding forward measures are exactly determined, we can estimate the exercise probabilities $\mathbb{P}_{T_{k}}\left[\left(CB(T_{0})>1\right)\right]$ under each forward measure by relying on a cumulant expansion for $\mathbb{P}_{T_{k}}\left[CB(T_{0})\right]$.

\section{The Wishart Libor Model}
The aim of this section is to illustrate a specific choice for the driving process $\Sigma$. As in the general setup, we specify the process under the terminal probability measure $\mathbb{P}_{T_N}$. The example we choose is the Wishart process, which was already presented in section \ref{Wis_disc}:

\begin{equation}
d\Sigma_t=(\Omega\Omega^{\top}+M\Sigma_t+\Sigma_tM^{\top})dt+\sqrt{\Sigma_t}dW_t^{T_N}Q+Q^{\top}dW_t^{{T_N}\top}\sqrt{\Sigma_t}.\label{wishart_ni}
\end{equation}

Here $W_t^{T_N}$ denotes a matrix Brownian motion, i.e. a $d\times d$ matrix of independent Brownian motions under the $\mathbb{P}_{N}$-forward probability measure. In the sequel we will write $W_t$ for notational simplicity.\\ 

In this section we show the impact of the relevant parameters on the implied volatility surface generated by vanilla options for a Libor model driven by a Wishart process. With the aim to investigate some complex movements of the implied volatility surface, we first compute the covariation between the Libor rate and its volatility: this covariation is a crucial quantity allowing for the so called skew effect on the smile, in perfect analogy with the leverage effect for vanilla options in the equity market.
 
\subsection{The skew of vanilla options}
We want to compute the covariation between the Libor rate and its volatility, so we proceed to derive the dynamics of the Libor rate in the Wishart  model. This may be done along the following steps: using the shorthand
\begin{align}
B_k:=B_{T_N-t}(u_k,u_{k+1})=-\psi_{T_N-t}(u_k)+\psi_{T_N-t}(u_{k+1}),
\end{align}
recall that we have:
\begin{align}
1+\Delta T L(t,T_k,T_{k+1})=\frac{B(t,T_k)}{B(t,T_{k+1})}=e^{A_k+\tr\left[B_k \Sigma_t\right]}.
\end{align}
In differential form, after dividing both sides by $L(t,T_k,T_{k+1})$ we have
\begin{align}
&\frac{dL(t,T_k,T_{k+1})}{L(t,T_k,T_{k+1})}\nonumber\\
&=\frac{1+\Delta T L(t,T_k,T_{k+1})}{L(t,T_k,T_{k+1})}\left([...]dt+\tr\left[B_k d\Sigma_t\right]\right).
\label{dynLibor}\end{align}

To preserve analytical tractability, we freeze the coefficients and approximate as follows:
\begin{align}
\frac{1+\Delta T L(t,T_k,T_{k+1})}{L(t,T_k,T_{k+1})}\approx \frac{1+\Delta T L(0,T_k,T_{k+1})}{L(0,T_k,T_{k+1})}=:C.
\label{frozen}\end{align}

\begin{proposition}
Under the assumption of frozen coefficients (\ref{frozen}), 
the conditional infinitesimal correlation between the Libor rate and its volatility cannot be negative and is given by 
\begin{align}
&d\left\langle  L(t,T_k,T_{k+1}), vol(L(t,T_k,T_{k+1}))\right\rangle\nonumber\\
&=\frac{\tr\left[B_kQ^{\top}QB_kQ^{\top}QB_k\Sigma\right]dt}{\sqrt{\tr\left[QB_k\Sigma B_k^{\top}Q^{\top}\right]}\sqrt{\tr\left[\Sigma B_kQ^{\top}QB_kQ^{\top}QB_kQ^{\top}QB_k \right]}}.
\end{align}
\label{propskew}\end{proposition}

\begin{proof} See Appendix. \end{proof}

From the previous formula we realize that the matrix $Q$ is responsible for the shape of the skew. We also have an indirect impact of the mean reversion speed matrix $M$ coming from the term $B_k$ which is the difference of two solutions of the Riccati ODE's \eqref{Riccati_psi} and \eqref{Riccati_phi}. The presence of $\Sigma$ suggests that in the present framework the skew is stochastic. What is more, it can only have positive sign.

\subsection{Numerical illustration with diagonal parameters}
The dynamics above show that the Wishart specification provides a very rich structure of the model. Since we want to get an understanding of the impact of different parameters we will look first at the case where all matrices are diagonal, which basically corresponds to a model driven by a two factor square root process (see e.g. \cite{dafgra11}).

We use the following set of parameters as a benchmark:
\begin{align}
\Sigma_0&=\left(\begin{array}{cc} 3.75&0\\ 0&3.45\end{array}\right),\quad M=\left(\begin{array}{cc} -0.3125*1.0e^{-003}&0\\ 0& -0.5000*1.0e^{-003} \end{array}\right),\nonumber\\
Q&=\left(\begin{array}{cc} 0.034&0\\ 0&0.0420\end{array}\right),\quad \kappa=3.\nonumber
\end{align}
The impact of the Gindikin parameter $\kappa$ is quite easy to understand: the process acts by influencing the overall level of the surface. This is due to the fact that the higher $\kappa$ the lower the probability that the process $\Sigma$ approaches 0. It is interesting to note that there is not only a level impact, but also a curvature effect, as we can see in Figure \ref{fig:1}.
\begin{center}
[Insert Figure \ref{fig:1} here]
\end{center}
Let us now look at the parameters along the diagonals of the matrices $M$ and $Q$. The following claims may be easily checked by looking at the SDE's satisfied by the elements of $\Sigma$ (see also \cite{article_DaFonseca4}). Note that we assumed all eigenvalues of $M$ to lie in the negative real line.
\begin{itemize}
\item For $\left|M_{11}\right|\nearrow(\searrow)$ the surface is shifted downwards (upwards).
\item For $\left|M_{22}\right|\nearrow(\searrow)$ the surface is shifted downwards (upwards).
\end{itemize}
The impact is more evident for OTM caplets with short maturities. This is due to the fact that as the process decreases (in matrix sense) the probability that caplets with short maturities are exercised is lowered more than the analogous probability for longer term caplets.
\begin{center}
[Insert Figure \ref{fig:2}  here]
\end{center}
We then consider the impact of $Q_{11},Q_{22}$. We have the following:
\begin{itemize}
\item As $Q_{11}\nearrow(\searrow)$ the surface is shifted upwards (downwards). In particular if we multiply $Q_{22}$ by a constant $c>1$, then the increment in the short term is higher for OTM than for ITM caplets. If $c<1$ then the decrease is higher for short term OTM caplets, which is intuitive, given the discussion above.
\item The same impacts, with different magnitudes, is observed also for $Q_{22}$.
\end{itemize}
\begin{center}
[Insert Figure \ref{fig:3}  here]
\end{center}

\subsection{The term structure of ATM implied volatilities for caplets}
\subsubsection{Diagonal parameters}
We proceed to consider the term structure of caplet implied volatilities. When the matrix $\Sigma_0$ is diagonal, the impact of the elements of $Q$ is the same: an increase in the absolute value of any element of $Q$ will result in a steeper term structure of ATM caplet volatilities.
\begin{center}
[Insert Figure \ref{fig:4}  here]
\end{center}
Considering a model where $\Sigma_0$ is a full matrix does not influence this result in a significant way.

\subsubsection{More complex adjustments: impact of off-diagonal elements}
To appreciate the flexibility of the Wishart framework, we focus now on the impact of the off-diagonal elements. We introduce off-diagonal elements in $M$ and $Q$ and look at the relative change in the short term smile (4 months) and the long term smile (32 months).
We introduce a fully populated matrix $\Sigma_0$ and look at the impact of $M_{12}$ and $M_{21}$. Our experiments show that there is a symmetry between the sign of $\Sigma_{0,12}$ (the initial value of $\Sigma_{12}$) and $M_{12}$, $M_{21}$. More precisely, the implied volatility changes are as in Table \ref{tab:ImpliedVolatilityChanges}.	
\begin{table*}[h]
\centering
	\begin{tabular}{ccc}
\multicolumn{3}{c}{  } \\ \hline
	 &$\Sigma_{0,12}>0$&$\Sigma_{0,12}<0$\\  \hline  \hline
	 $M_{12}>0$&Increase&Decrease\\
	 $M_{12}<0$&Decrease&Increase \\
	 &&\\   \hline
   $M_{21}>0$&Increase&Decrease\\
	 $M_{21}<0$&Decrease&Increase \\ \hline
	\end{tabular}
	\caption{Implied volatility changes: relation between $\Sigma_{0,12}$ and $M_{12},M_{21}$.}
	\label{tab:ImpliedVolatilityChanges}
\end{table*}

The reason for this symmetry is to be looked for in the drift part of the dynamics of the single elements of the matrix process $\Sigma$.\\

Next we look at the impact of $Q_{12}$ and $Q_{21}$. To this end we model $Q$ as a symmetric matrix and set $Q_{21}=Q_{12}=\rho\sqrt{Q_{11}Q_{22}}$ for a real parameter $\rho$. Also in this case we recognize two main shapes of the adjustment that we denote by $S_1,S_2$.

\begin{table*}[h]
\centering
	\begin{tabular}{ccc}
\multicolumn{3}{c}{  } \\ \hline
	 &$\Sigma_{0,12}>0$&$\Sigma_{0,12}<0$\\ \hline \hline
	 $\rho>0$&$S_1$&$S_2$\\ \hline
	 $\rho<0$&$S_2$&$S_1$\\ \hline
	\end{tabular}
	\caption{Implied volatility changes: relation between $\Sigma_{0,12}$ and $\rho$.}
	\label{tab:ImpliedVolatilityChanges_rho}
\end{table*}

We now proceed to perform other numerical tests which will show that our modeling framework has a certain degree of flexibility. For these tests we set:
\begin{align}
M&=\left(\begin{array}{cc} -0.3125*1.0e^{-003} &0\\ 0& -0.5000*1.0e^{-003} \end{array}\right),\nonumber\\
Q&=\left(\begin{array}{cc} 0.02&\rho \sqrt{Q_{11}Q_{22}}\\ \rho \sqrt{Q_{11}Q_{22}}&0.0420\end{array}\right),\quad \kappa=3,\nonumber
\end{align}
so basically $M$ is parametrized  as before but $Q$ is symmetric and equiped with a parameter $\rho$ which summarizes the information on the off-diagonal elements. We require $\Sigma_0=\Sigma_\infty$, where $\Sigma_\infty$ is given by the solution of the Lyapunov equation \eqref{Lyap}. After that we perturbate $\Sigma_0$ in order to include off-diagonal elements and set $\Sigma_{0,12}=\Sigma_{0,21}=2$.
We have a good degree of control on the term structure of ATM implied volatilities. In particular, we may have larger percentage shifts in the long-term w.r.t. the short-term ATM implied volatility, or, for $\rho=-0.6$ we may even reproduce a situation where the short term ATM implied volatility increases whereas the long-term ATM implied volatility decreases.
\begin{center}
[Insert Figure \ref{fig:5}  here]
\end{center}

If we adopt the same kind of parametrization for the matrix $M$ by introducing a second parameter $\rho_2$, then we have further flexibility because we can impose many different combinations of $\rho$ and $\rho_2$. For example, Figure \ref{fig:9} shows that we are able  to isolate an effect on the term structure of ATM implied volatility: in fact we have a moderate change for ITM caplets while OTM caplets are practically unchanged, but the shape of the term structure of ATM implied volatility is modified in a significant way.
\begin{center}
[Insert Figure \ref{fig:9}  here]
\end{center}

Finally, just for illustrative purposes we report a prototypical Caplet volatility surface generated by the model.
\begin{center}
[Insert Figure \ref{fig:7}  here]
\end{center}
As far as Swaptions are concerned an example of ATM implied volatility surface for different expiries and underlying swap lengths is given below.
\begin{center}
[Insert Figure \ref{fig:8}  here]
\end{center}

\section{The Pure Jump Libor Model}
Finally, in this section, we would like to provide a second example for the driving process $\Sigma$, so as to let the reader appreciate the degree of generality of this framework.
As in the general setup, we specify the process under the terminal probability measure $\mathbb{P}_{T_N}$. The example we choose is a matrix compound Poisson process, which was already presented in section \ref{pure_jump}:
\begin{align}
d\Sigma_t=M\Sigma_t+\Sigma_tM^{\top}+dL_t^{\mathbb{P}_{T_N}}.
\end{align}
All assumptions presented in section \ref{pure_jump} are in order. More specifically, we assume that $L_t^{\mathbb{P}_{T_N}}$ is a compound Poisson process with constant intensity $\lambda$ and jump distribution taking values in $S_d^{++}$. As a specific example of jump distribution we choose the standard Wishart distribution. By recalling the results in section \ref{pure_jump} we have that the solution for the function $\psi_\tau(u)$ is
\begin{align}
\psi_\tau(u)=e^{M^{\top}\tau}ue^{M\tau},
\end{align}
whereas for $\phi_\tau(u)$ we have
\begin{align}
\phi_\tau(u)=-\lambda\int_{0}^{\tau}{\det\left(I_d+2e^{M^{\top}s}ue^{Ms}\mathcal{Q}\right)^{-\frac{n}{2}}ds}+\lambda \tau.
\end{align}
In concrete pricing applications, the computation of the solution for $\phi_\tau(u)$ implies a numerical integration with respect to the time dimension. This numerical integration has an impact on the performance of the model which turns out to be slower than the Wishart Libor model. For illustrative purposes, we report an example for an implied volatility surface for caplets generated by the compound Poisson Libor model with central Wishart distributed jumps.
The mean reversion matrix $M$ and the jump intensity $\lambda$ are given by:
\begin{align}
M&=\left(\begin{array}{cc} -0.0550&0\\ 0& -0.1760\end{array}\right),\nonumber\\
\lambda&=0.1\nonumber.
\end{align}
As far as the jump size distribution is concerned, the parameters are the following:
\begin{align}
\mathcal{Q}&=\left(\begin{array}{cc} 0.27&0\\ 0& 0.05\end{array}\right),\nonumber\\
n&=3.1, \nonumber
\end{align}
and the initial state of the process is
\begin{align*}
\Sigma_0&=\left(\begin{array}{cc} 1.875&0.6\\ 0.6& 1.275\end{array}\right).
\end{align*}

\begin{center}
[Insert Figure \ref{fig:10}  here]
\end{center}

\section{Conclusions}
In this paper we presented an extension of the approach of \cite{article_KRPT2009} to the more general setting of affine processes on positive definite matrices. We showed that their methodology may be adapted to this state space in a straightforward way. What is more, it is possible to efficiently price European swaptions in this multi-factor setting by means of a cumulant expansion due to \cite{article_CollinDufresneGoldstein2002}. In doing so we are in front of a setting which is potentially able to capture correlation effects which can not be described by a single-factor framework. We provided numerical examples for the Wishart Libor model, where the introduction of off-diagonal elements gives rise to new possibilities in the control of the shape of the implied volatility surface.\\

Our contribution may be seen as a starting point for a description of market models in this state space, in consequence we believe that there are many possible directions for future research. An example is given by the problem of calibrating this family of models to real market data. As the structure of the products in the fixed-income market suggests, even in the plain vanilla case, we expect the objective function that should be minimized in the calibration procedure to be quite involved. Yet, some calibration results were obtained on equity derivatives in \cite{dafgra11} for Wishart based models so a calibration using interest rates derivatives might be feasible. Certainly, it will be a delicate issue and may constitute an interesting contribution by its own. Once the model is calibrated on vanillas, one could then further investigate the performance of the model on more exotic structures, like e.g. Bermudan swaptions and barrier options. Theses issues are left for future work.

\clearpage

\section*{Appendix A: proofs}

\subsection*{Proof of Theorem \ref{TeoremaMartingale}}
For all $ u \in \mathcal{I}_{T}$ we have
\begin{equation}
\mathbb{E}\left[ M^{u}_{T}\right]=\mathbb{E}\left[e^{-\tr\left[ u\Sigma_{T}\right]} \right]<\infty,\nonumber
\end{equation}
and by the affine property we obtain
\begin{align}
\mathbb{E}\left[ M^{u}_{T}|\mathcal{F}_{t}\right]&=\mathbb{E}\left[ \exp\left\{-\phi_{T-T}(u)-\tr\left[\psi_{T-T}(u)\Sigma_{T} \right]\right\}|\mathcal{F}_{t}\right]\nonumber\\
&=\mathbb{E}\left[\exp\left\{-\tr\left[u\Sigma_{T} \right] \right\} |\mathcal{F}_{t}\right]\nonumber\\
&=\exp\left\{ -\phi_{T-t}(u)-\tr\left[ \psi_{T-t}(u)\Sigma_{t}\right]\right\}=M^{u}_{t},\nonumber
\end{align}
hence the process is a martingale. Now we show that $M_{t}^{u}>1$. Recall that by assumption $u\in \mathcal{I}_{T}\cap S_{d}^{--}$ and 
\begin{equation}
M^{u}_{t}=\mathbb{E}\left[\exp\left\{-\tr\left[u\Sigma_{T} \right] \right\} |\mathcal{F}_{t}\right],\nonumber\\
\end{equation}
so that if $-\tr\left[u\Sigma_{T} \right]>0$ a.s. then we are done.
We proceed as in \cite{article_gou02} and apply the singular value decomposition to the negative definite matrix $u$, i.e. $u$ may be written as:
\begin{equation}
u=\sum_{i=1}^{n}{\lambda_{i}u_{i}u_{i}^{\top}}\nonumber
\end{equation}
where $\lambda_{i}$ are the eigenvalues of $u$ and $u_i$ are the eigenvectors. By assumption $\Sigma_{T}$ takes values in $S_{d}^{++}$, hence
\begin{align}
-\tr\left[u\Sigma_{T} \right]&= -\tr\left[ \sum_{i=1}^{n}{\lambda_{i}u_{i}u_{i}^{\top}} \Sigma_{T}\right] \nonumber\\
&=-\sum_{i=1}^{n}{\lambda_{i} \tr\left[u_{i}u_{i}^{\top}\Sigma_{T}\right]} \nonumber\\
&=-\sum_{i=1}^{n}{\lambda_{i}u_{i}^{\top}\Sigma_{T}u_{i}}>0 
\end{align}
as we wanted.\endproof

\subsection*{Proof of Proposition \ref{propgamma}}
We follow closely the proof in \cite{article_KRPT2009}. By assumption, initial Libor rates are strictly positive, then
\begin{equation}
\frac{B(0,T_{1})}{B(0,T_{N})} > \frac{B(0,T_{2})}{B(0,T_{N})} > ... > \frac{B(0,T_{N})}{B(0,T_{N})}=1.
\end{equation}
Recall that we have
\begin{equation}
\mathbb{E}\left[ e^{-\tr\left[u_1\Sigma_{T}\right]}\right]=M_{0}^{u_{1}}=\exp\left\{ -\phi_{T}(u_{1})-\tr\left[\psi_{T}(u_{1})\Sigma_{0} \right]\right\}=\frac{B(0,T_{1})}{B(0,T_{N})}.
\end{equation}
By the definition of $\gamma_{\Sigma}$ in (\ref{gammasigma}), we have that if $\gamma_{\Sigma}=\infty$ then we are done, else we can claim that there exists an $\epsilon >0$ such that $\gamma_{\Sigma}-\epsilon > \frac{B(0,T_{1})}{B(0,T_{N})}$. Then we can find a matrix $ \tilde{u}$ s.t.
\begin{equation}
\mathbb{E}\left[ e^{-\tr\left[\tilde{u}\Sigma_{T}\right]}\right]>\gamma_{\Sigma}-\epsilon > \frac{B(0,T_{1})}{B(0,T_{N})}.
\end{equation}
In analogy with \cite{article_KRPT2009} we introduce the function
\begin{align}
&f:\left[0,1\right]\rightarrow\mathbb{R}_{\geq 0}\nonumber\\
&\xi \rightarrow \mathbb{E}\left[ e^{-\tr\left[\xi \tilde{u}\Sigma_{T}\right]}\right]
\end{align}
and we want to show that $f$ is continuous. First, since $\Sigma \in S_{d}^{++}$ and $u \in S_{d}^{--}$ we have that if $u \prec v$ then $-\tr\left[u\Sigma_{T}\right]>-\tr\left[v\Sigma_{T}\right]$, hence by monotone convergence we can conclude that $f$ is increasing. We now introduce an increasing sequence $(a_{n})_{n \in \mathbb{N}} \nearrow 1$ and apply Fatou's lemma to obtain
\begin{equation}
\liminf_{n \rightarrow \infty}\mathbb{E}\left[ e^{-\tr\left[a_{n} \tilde{u}\Sigma_{T}\right]}\right] \geq
\mathbb{E}\left[\liminf_{n \rightarrow \infty} e^{-\tr\left[a_{n} \tilde{u}\Sigma_{T}\right]}\right]=\mathbb{E}\left[ e^{-\tr\left[\tilde{u}\Sigma_{T}\right]}\right],
\nonumber\end{equation}
implying that $f$ is lower semi-continuous. Since $f$ is also increasing we have that $f$ is continuous. Now $f(0)=1$ and $f(1)>\frac{B(0,T_{1})}{B(0,T_{N})}$, hence there exist some numbers $0= \xi_{N} < \xi_{N-1} < ... < \xi_{1}<1$ such that
\begin{equation}
f\left(\xi_{k}\right)=M_{0}^{\xi_{k}\tilde{u}}=\frac{B(0,T_{k})}{B(0,T_{N})}, \hspace{10mm}\forall k \in \left\{ 1,...,N\right\}.
\nonumber\end{equation}
By setting $u_{k}=\xi_{k}\tilde{u}$ (for $k=1,...,N-1$) we obtain a sequence of matrices $u_{k}\prec u_{k+1}, u_{k}-u_{k+1}\in S_d^{--}$ which allows us to fit the initial tenor structure of Libor rates as desired.
Finally, we apply Proposition 1 and Lemma 3.2 (ii) in \cite{article_Cuchiero} in order to obtain the last sentence of the Proposition \ref{propgamma}.
\endproof

\subsection*{Proof of Proposition \ref{propskew}}
In this section we proceed as in the proof of Proposition 4.1 in \cite{article_DaFonseca1}. Recall that $W_t$ is a shorthand for $W_t^{T_N}$. From (\ref{dynLibor}) it follows that
\begin{align}
\frac{dL(t,T_k,T_{k+1})}{L(t,T_k,T_{k+1})}&=C\left((...)dt+2\sqrt{\tr\left[QB_k\Sigma B_k^{\top}Q^{\top}\right]}\left(\frac{\tr\left[QB_k\sqrt{\Sigma}dW_t\right]}{\sqrt{\tr\left[QB_k\Sigma B_k^{\top}Q^{\top}\right]}}\right)\right)\nonumber\\
&:=C\left((...)dt+2\sqrt{\tr\left[QB_k\Sigma B_k^{\top}Q^{\top}\right]}d\tilde{W}_t\right),
\end{align}
where $C$ was defined in \eqref{frozen} and the scalar noise driving the factor process may be derived as follows:
\begin{align}
&d\tr\left[QB_k\Sigma_t B_kQ^{\top}\right]=\left(\tr\left[QB_k\beta Q^{\top}QB_kQ^{\top}\right]+2\tr\left[QB_kM\Sigma_t B_kQ^{\top}\right]\right)dt\nonumber\\
&+2\tr\left[QB_k\sqrt{\Sigma_t}dW_tQB_kQ^{\top}\right]\nonumber\\
&=(...)dt+2\sqrt{\tr\left[\Sigma_t B_kQ^{\top}QB_kQ^{\top}QB_kQ^{\top}QB_k\right]}\frac{\tr\left[QB_kQ^{\top}QB_k\sqrt{\Sigma_t}dW_t\right]}{\sqrt{\tr\left[\Sigma B_kQ^{\top}QB_kQ^{\top}QB_kQ^{\top}QB_k \right]}}\nonumber\\
&:=(...)dt+2\sqrt{\tr\left[\Sigma_t B_kQ^{\top}QB_kQ^{\top}QB_kQ^{\top}QB_k\right]}dZ_t.
\end{align}
The covariation between the noise of the Libor rate and its volatility is then given by
\begin{align}
\left\langle d\tilde{W}_t,dZ_t\right\rangle&=\left\langle \frac{\tr\left[QB_k\sqrt{\Sigma_t}dW_t\right]}{\sqrt{\tr\left[QB_k\Sigma_t B_k^{\top}Q^{\top}\right]}},\frac{\tr\left[QB_kQ^{\top}QB_k\sqrt{\Sigma_t}dW_t\right]}{\sqrt{\tr\left[\Sigma_t B_kQ^{\top}QB_kQ^{\top}QB_kQ^{\top}QB_k \right]}} \right\rangle\nonumber\\
&=\frac{\left(\sum_{p,q,r,s}{Q_{pq}B_{qr}\sqrt{\Sigma}_{rs}dW_{sp}}\right)\left(\sum_{a,b,c,d,e,f,g}{Q_{ab}B_{bc}Q^{\top}_{cd}Q_{de}B_{ef}\sqrt{\Sigma}_{fg}dW_{ga}}\right)}{\sqrt{\tr\left[QB_k\Sigma B_k^{\top}Q^{\top}\right]}\sqrt{\tr\left[\Sigma B_kQ^{\top}QB_kQ^{\top}QB_kQ^{\top}QB_k \right]}}\nonumber\\
&=\frac{\sum_{a,b,c,d,e,f,g,q,r}{B_{fe}Q^{\top}_{ed}B_{cb}Q^{\top}_{ba}Q_{aq}B_{qr}\sqrt{\Sigma}_{rg}\sqrt{\Sigma}_{gf}}dt}{\sqrt{\tr\left[QB_k\Sigma B_k^{\top}Q^{\top}\right]}\sqrt{\tr\left[\Sigma B_kQ^{\top}QB_kQ^{\top}QB_kQ^{\top}QB_k \right]}}\nonumber\\
&=\frac{\tr\left[B_kQ^{\top}QB_kQ^{\top}QB_k\Sigma_t\right]dt}{\sqrt{\tr\left[QB_k\Sigma_t B_k^{\top}Q^{\top}\right]}\sqrt{\tr\left[\Sigma_t B_kQ^{\top}QB_kQ^{\top}QB_kQ^{\top}QB_k \right]}}.\nonumber
\end{align}

Now we turn on the positivity of the skew. With the notation in the proof of Proposition \ref{propgamma}, from $\xi_k>\xi_{k+1}$ we have $u_k\prec u_{k+1}$ and then $B_k\in S_d^+$. In all terms in the numerator and the denominator we recognize congruent transformations of matrices in $S_d^+$ which leave the signs of the eigenvalues unchanged. The self-duality of $S_d^+$ allows us to claim that all traces are positive, hence we are done.
\endproof

\section*{Appendix B: the characteristic function}
With the purpose of pricing caplets, we need to have a more explicit form for the characteristic function appearing in Proposition \ref{price_caplet}. Once we have this expression we can plug in the functions $\phi_{\tau}(u)$ and $\psi_{\tau}(u)$ to obtain a closed form solution. The pricing problem will be then solved via FFT. Recall that we are considering the following expectation:
\begin{align}
\varphi(v)&=\mathbb{E}^{\mathbb{P}_{T_{k+1}}}\left[e^{i\left( v-\left(\alpha +1 \right)i\right)\left(A_{k}+\tr\left[B_{k}\Sigma_{T_{k}}\right]\right)}\right]\nonumber\\
&=e^{i\left( v-\left(\alpha +1 \right)i\right)A_{k}}\mathbb{E}^{\mathbb{P}_{T_{k+1}}}\left[\exp\left\{\tr\left[\underbrace{i\left( v-\left(\alpha +1 \right)i\right) B_{k}}_{u}\Sigma_{T_{k}}\right]\right\}\right]
\end{align}
where
\begin{align}
&A_{k}:=-\phi_{T_{N}-T_{k}}(u_{k})+\phi_{T_{N}-T_{k}}(u_{k+1});\nonumber\\
&B_{k}:=-\psi_{T_{N}-T_{k}}(u_{k})+\psi_{T_{N}-T_{k}}(u_{k+1}).
\end{align}
As we computed the shape of the function $\phi_{\tau}(u)$ and $\psi_{\tau}(u)$ under the $\mathbb{P}_{T_{N}}$-forward measure, we need to switch from the $\mathbb{P}_{T_{k+1}}$ to the $\mathbb{P}_{T_{N}}$-forward measure:
\begin{align}
&e^{i\left( v-\left(\alpha +1 \right)i\right)A_{k}}\mathbb{E}^{\mathbb{P}_{T_{k+1}}}\left[e^{\tr\left[u\Sigma_{T_{k}}\right]}\right]\nonumber\\
&=e^{i\left( v-\left(\alpha +1 \right)i\right)A_{k}}\mathbb{E}^{\mathbb{P}_{T_{N}}}\left[\frac{\partial \mathbb{P}_{T_{k+1}}}{\partial \mathbb{P}_{T_{N}}}e^{\tr\left[u\Sigma_{T_{k}}\right]}\right]\nonumber\\
&=e^{i\left( v-\left(\alpha +1 \right)i\right)A_{k}}\mathbb{E}^{\mathbb{P}_{T_{N}}}\left[\frac{M^{u_{k+1}}_{T_k}}{M^{u_{k+1}}_{0}}e^{\tr\left[u\Sigma_{T_{k}}\right]}\right],\label{8.3}
\end{align}
where the last equation follows from \eqref{change_measure}. Let us focus on the expectation which becomes:
\begin{align}
&\mathbb{E}^{\mathbb{P}_{T_{N}}}\Bigg[\exp\Big\{ -\phi_{T_{N}-T_{k}}(u_{k+1})-\tr\left[\psi_{T_{N}-T_{k}}(u_{k+1})\Sigma_{T_{k}}\right]\Big. \Bigg. \nonumber\\
&\Bigg.\Big.+\phi_{T_{N}}(u_{k+1})+\tr\left[\psi_{T_{N}}(u_{k+1})\Sigma_0 \right]+\tr\left[u\Sigma_{T_{k}} \right]\Big\}\Bigg]\nonumber\\
&=\exp\Big\{ -\phi_{T_{N}-T_{k}}(u_{k+1})+\phi_{T_{N}}(u_{k+1})+\tr\left[\psi_{T_{N}}(u_{k+1})\Sigma_0\right] \Big\}\nonumber\\
&\times\mathbb{E}^{\mathbb{P}_{T_{N}}}\Bigg[e^{\tr\left[\left(-\psi_{T_{N}-T_{k}}(u_{k+1})+u \right)\Sigma_{T_{k}} \right]}\Bigg]\nonumber\\
&=\exp\Big\{ -\phi_{T_{N}-T_{k}}(u_{k+1})+\phi_{T_{N}}(u_{k+1})+\tr\left[\psi_{T_{N}}(u_{k+1})\Sigma_0\right]\Big.\nonumber\\
&\Big. -\phi_{T_{k}}\Big(-\psi_{T_{N}-T_{k}}(u_{k+1})+u\Big) -\tr\left[\psi_{T_{k}}\Big( -\psi_{T_{N}-T_{k}}(u_{k+1})+u\Big)\Sigma_0 \right]\Big\}.
\end{align}
Now, recalling the previous terms in front of the expectation in \eqref{8.3}, we obtain the final expression which is

\begin{align*}
\exp\Bigg\{ &i(v-(\alpha+1)i)\Big(\overbrace{-\phi_{T_{N}-T_{k}}(u_{k})+\phi_{T_{N}-T_{k}}(u_{k+1})}^{A_{k}} \Big)\Bigg.\nonumber\\
&-\left. \phi_{T_{N}-T_{k}}(u_{k+1})+\phi_{T_{N}}(u_{k+1})\right.\nonumber\\
&\left.-\phi_{T_{k}}\Big(-\psi_{T_{N}-T_{k}}(u_{k+1})+\underbrace{i(v-(\alpha+1)i)\Big(\overbrace{-\psi_{T_{N}-T_{k}}(u_{k})+\psi_{T_{N}-T_{k}}(u_{k+1})}^{B_{k}} \Big)}_{u} \Big)\right.\nonumber\\
&\left.+\tr\left[\psi_{T_{N}}(u_{k+1})\Sigma_0\right] \right.\nonumber\\
&\Bigg. -\tr\Bigg[\psi_{T_{k}}\Bigg(-\psi_{T_{N}-T_{k}}(u_{k+1})+\underbrace{i(v-(\alpha+1)i)\overbrace{\Big(-\psi_{T_{N}-T_{k}}(u_{k})+\psi_{T_{N}-T_{k}}(u_{k+1})\Big)}^{B_{k}}}_{u} \Bigg)\Sigma_0 \Bigg] \Bigg\}.
\end{align*}

\section*{Figures}

\begin{figure}[H]
\centering
\includegraphics[scale=0.14]{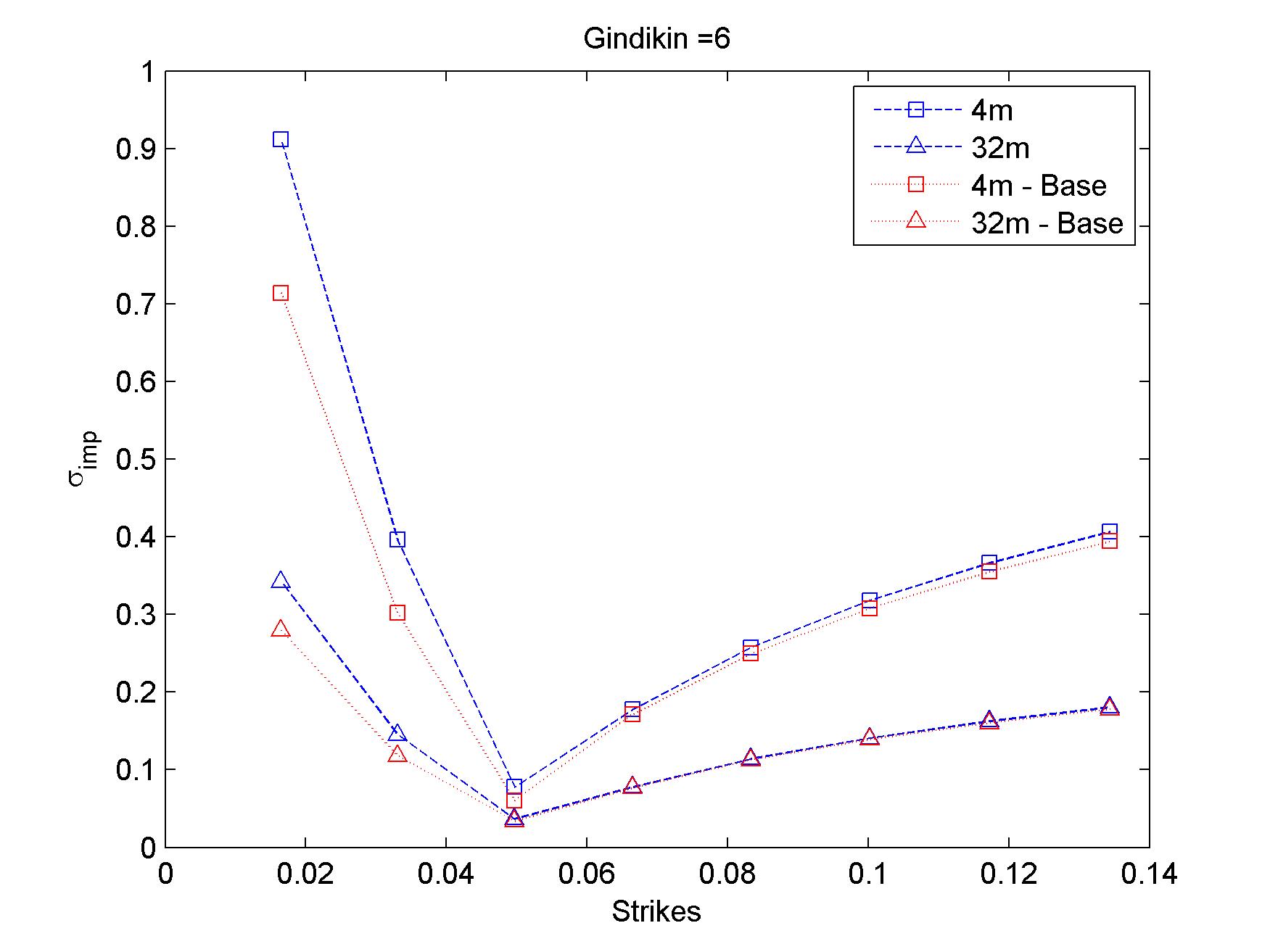}
\caption{Doubling $\kappa$ with respect to the basic case causes an upward shift of the surface. The plot represents the two smiles (4 months and 32 months) for the basic ($\kappa =3$) and the modified case ($\kappa = 6$).}\label{fig:1}
\end{figure}

\begin{figure}[H]
  \centering
  \includegraphics[scale=0.14]{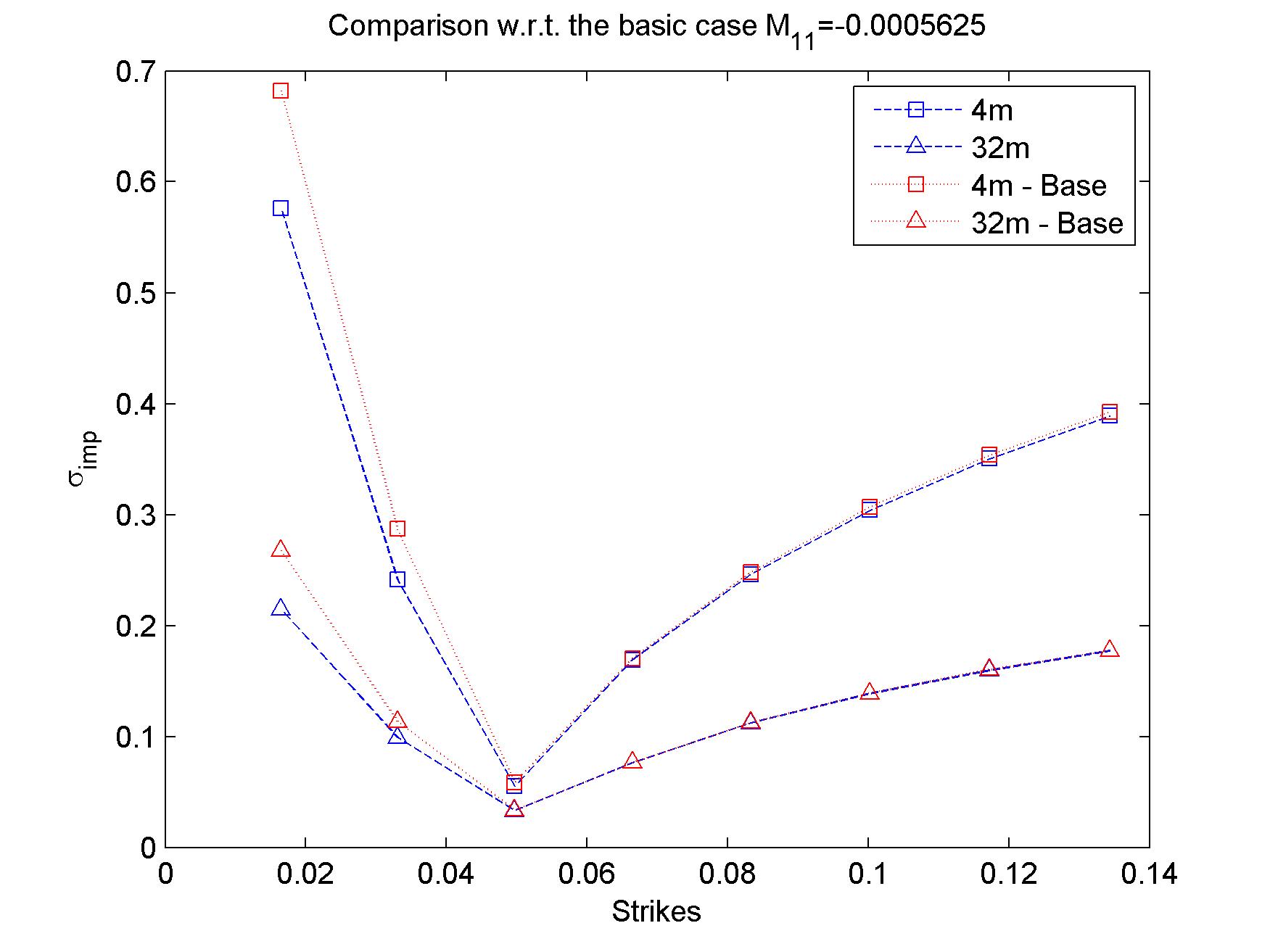}
  \caption{Impact of $M_{11}$. $M_{11}$ is negative and the present image shows the effects on the two smiles (4 months and 32 months) we obtain when we multiply it by a constant $c=1.8$}
  \label{fig:2}
\end{figure}

\begin{figure}[H]
  \centering
  \includegraphics[scale=0.18]{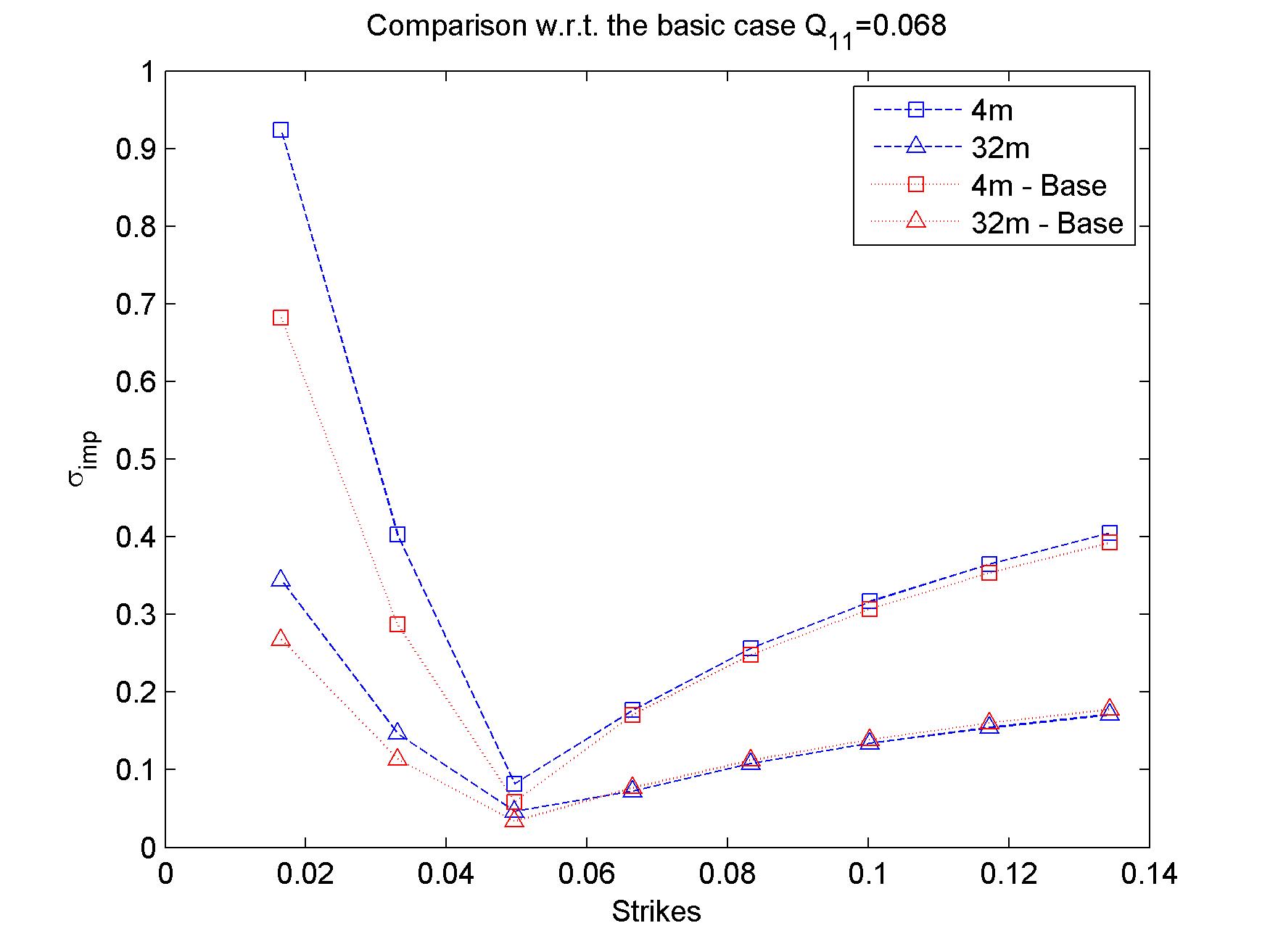}
  \caption{Impact of $Q_{11}$. $Q_{11}$ is positive and the present image shows the effects on the two smiles (4 months and 32 months) we obtain when we multiply it by a constant $c=2$}
  \label{fig:3}
\end{figure}

\begin{figure}[H]
  \centering
  \includegraphics[scale=0.18]{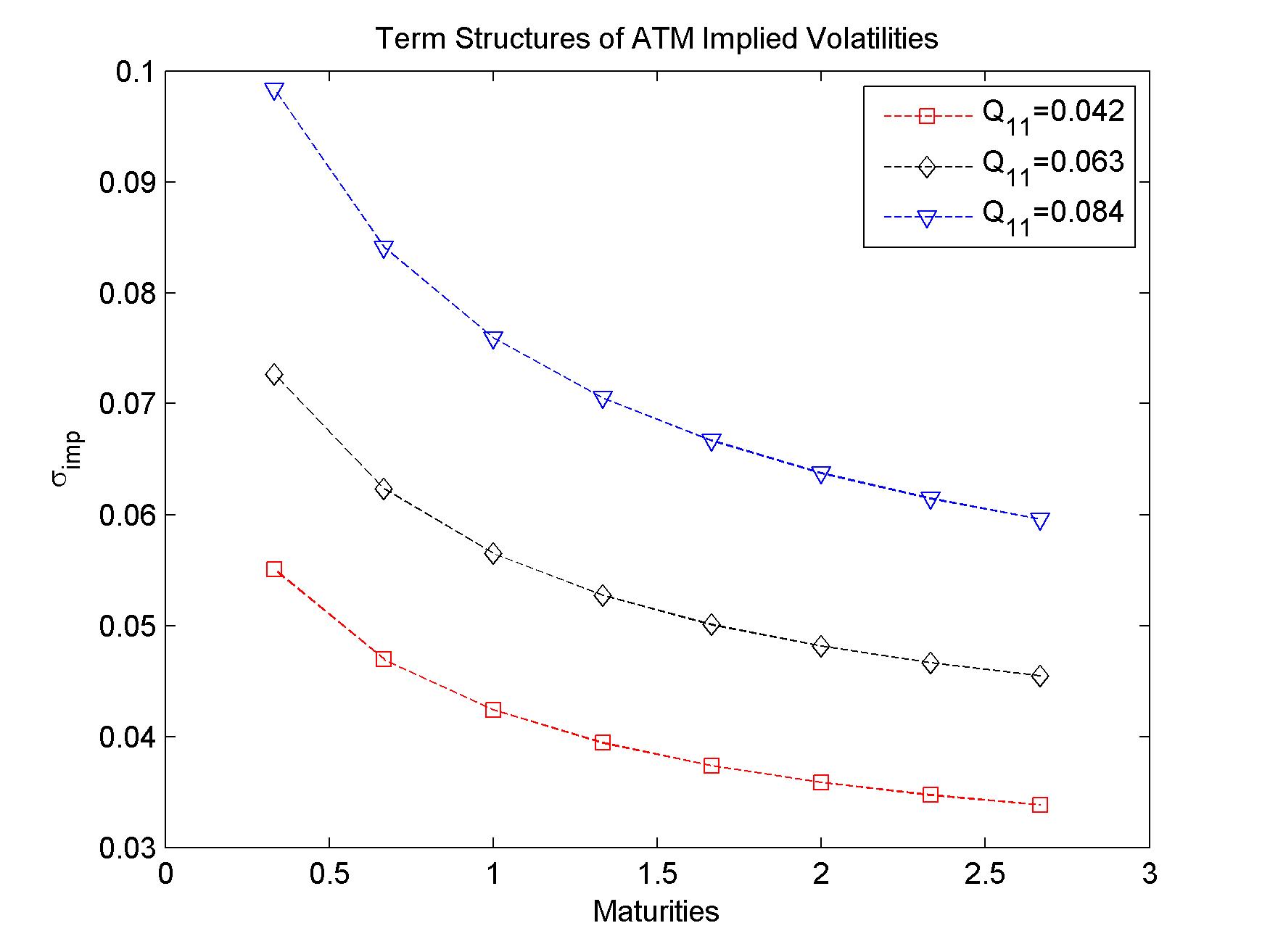}
  \caption{Impact of $Q$ on the term structures of ATM implied volatilities. Here we consider $Q_{11}$ and multiply its value by a constant $c=1,1.5,2$ so as to get the values in the legend.}
  \label{fig:4}
\end{figure}

\begin{figure}[H]
  \centering
  \subfloat{\label{fig:fig21}\includegraphics[width=0.5\textwidth]{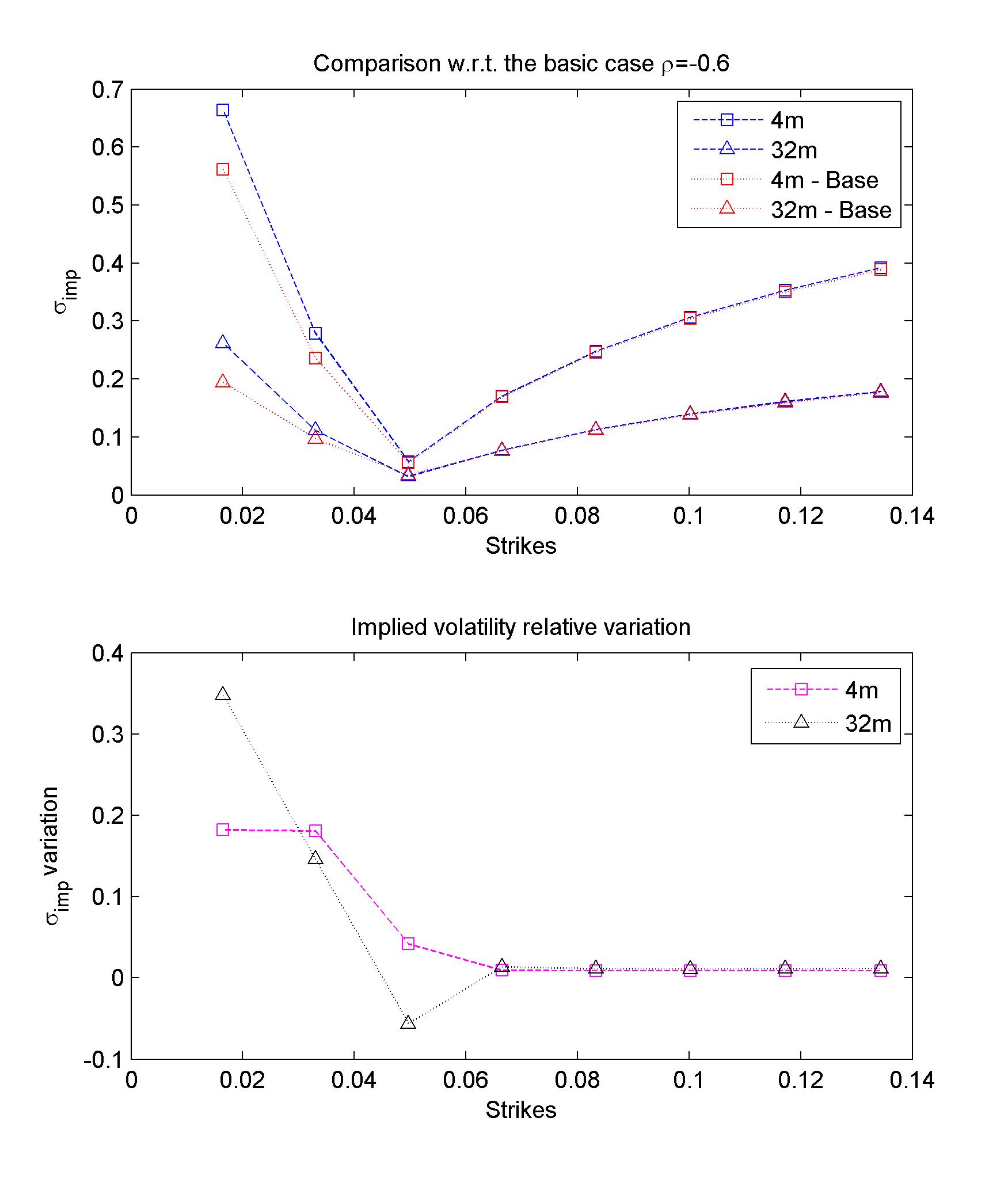}}                
  \subfloat{\label{fig:fig22}\includegraphics[width=0.5\textwidth]{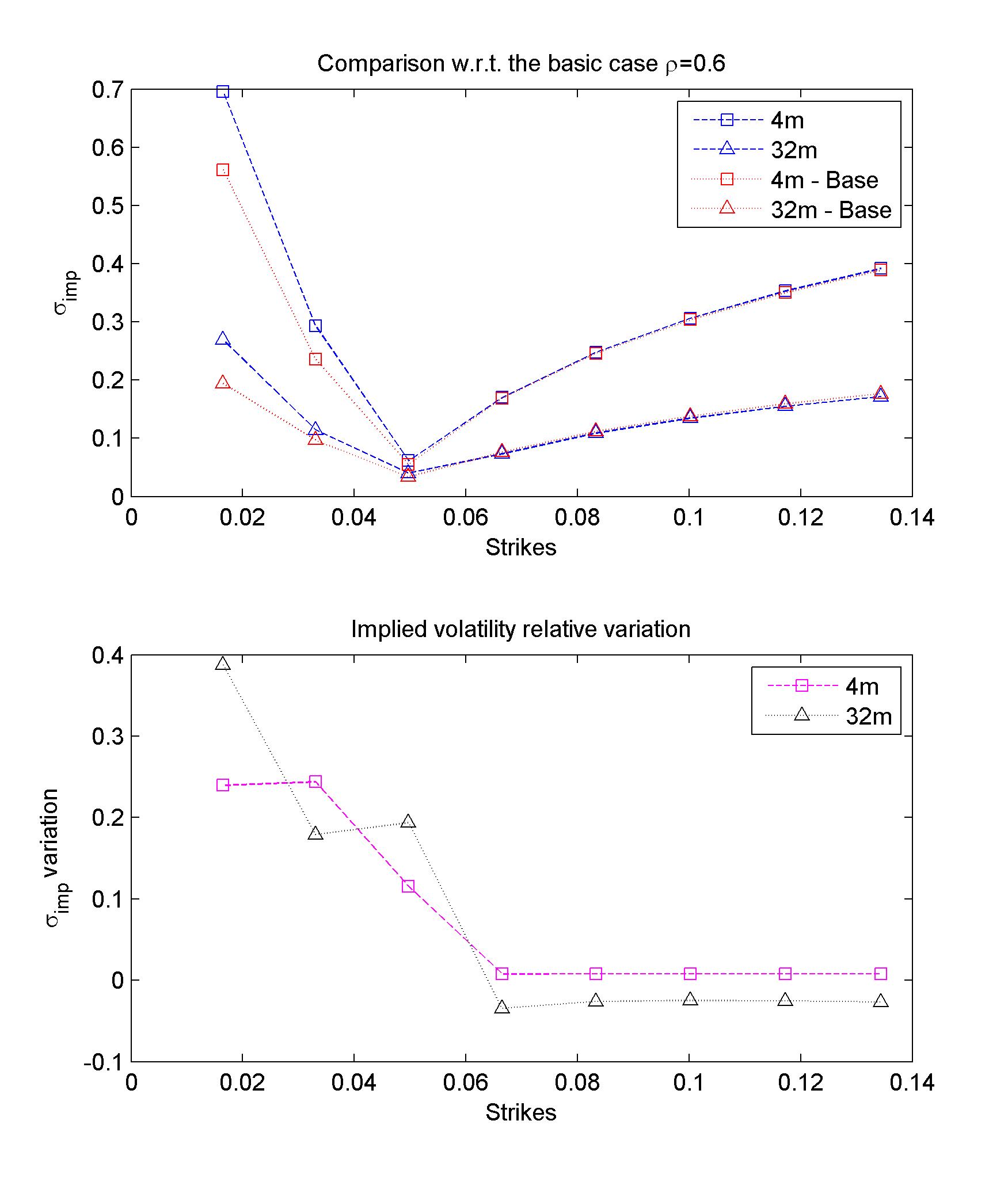}}               
\caption{The images above highlight the flexibility of the Wishart Libor model. We are able to impose different patterns to the term structure of ATM implied volatility. On the top we have the smiles and on the bottom we observe the relative changes of the smiles, i.e. for every point of the smiles we calculate the quantity $\left(\sigma^{imp}_{final}-\sigma^{imp}_{initial}\right)/\sigma^{imp}_{initial}$. Notice in particular the situation on the left side, where we observe around 5\% (ATM) an increase of the short term smile and a decrease on the long term.}
  \label{fig:5}
\end{figure}

\begin{figure}[H]
\centering
  \includegraphics[scale=0.18]{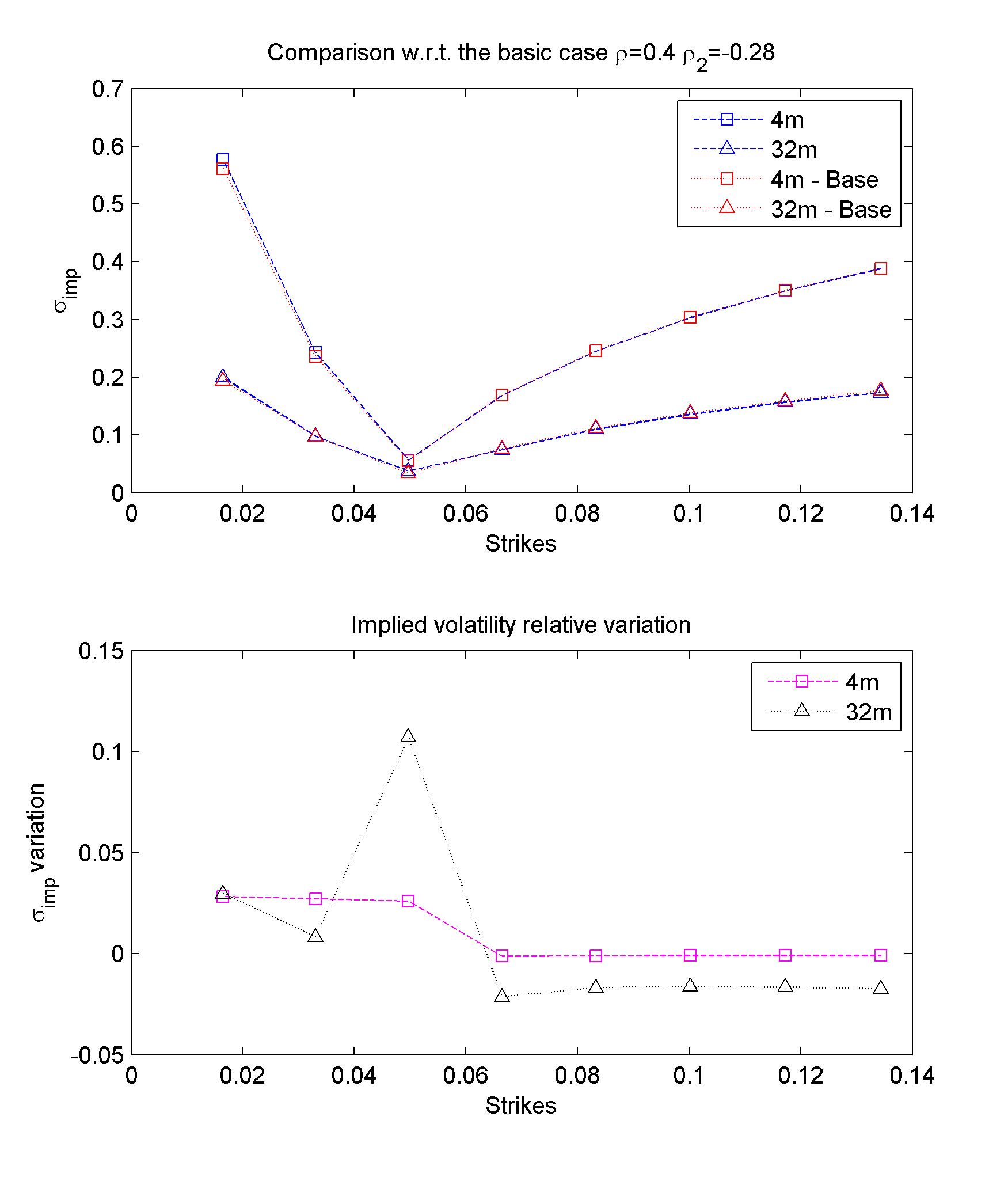}
  \caption{Impact on the implied volatility surface when both $M$ and $Q$ are parametrized as symmetric matrices. Notice the level around 5\%, corresponding to ATM. This shows that if we parametrize both $M$ and $Q$ via $\rho,\rho_2$ we have a flexible setting which is controlled just by two parameters that allow us to perform different combinations. In particular $\rho$ and $\rho_2$ have opposite impacts in the present example ($\rho>0$ whereas $\rho_2<0$), meaning that we have a good degree of control.}
  \label{fig:9}
\end{figure}

\begin{figure}[H]
\centering
  \includegraphics[scale=0.18]{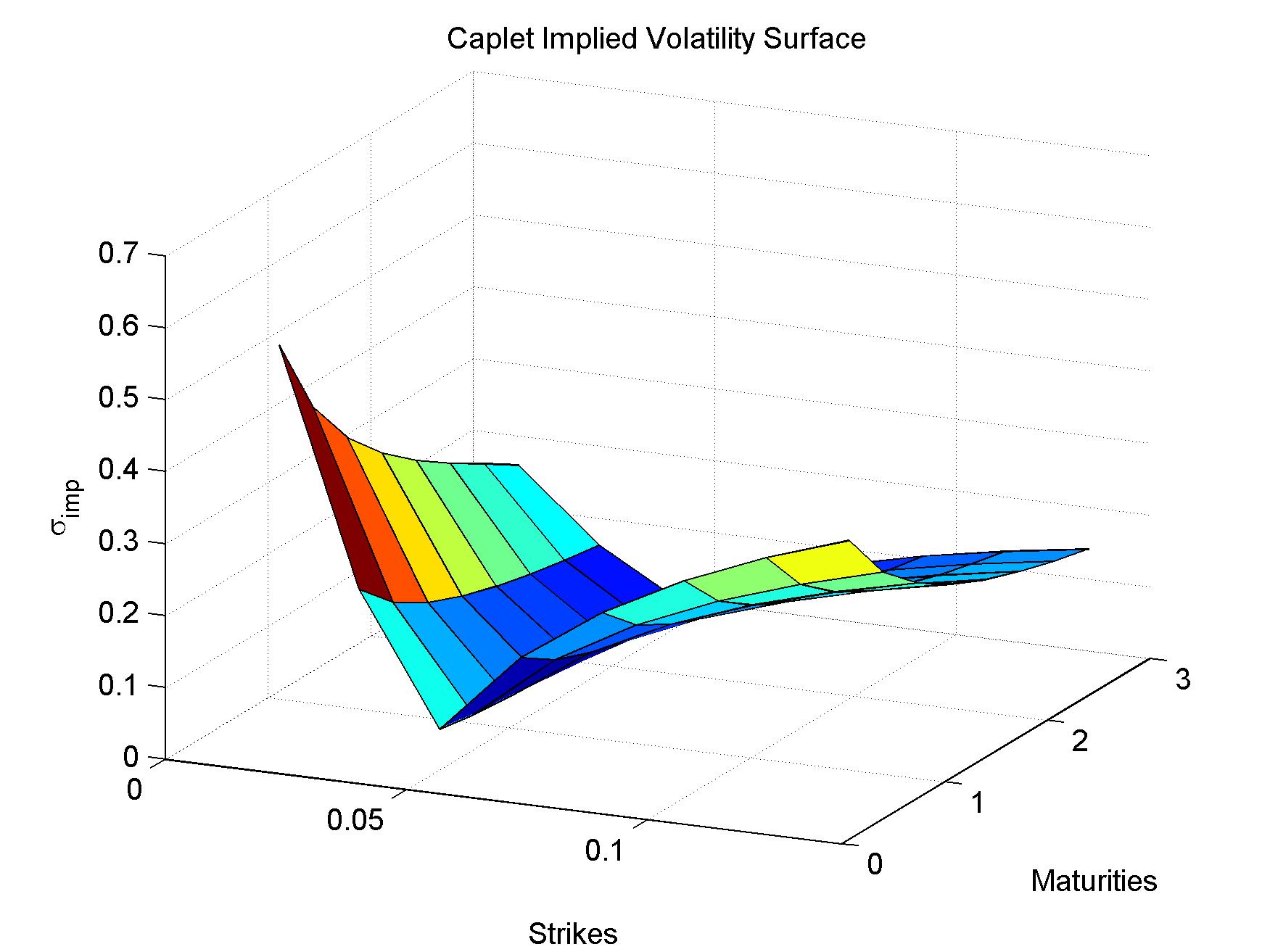}
  \caption{Caplet Implied Volatility Surface generated by the Wishart Libor model}
  \label{fig:7}
\end{figure}

\begin{figure}[H]
\centering
  \includegraphics[scale=0.18]{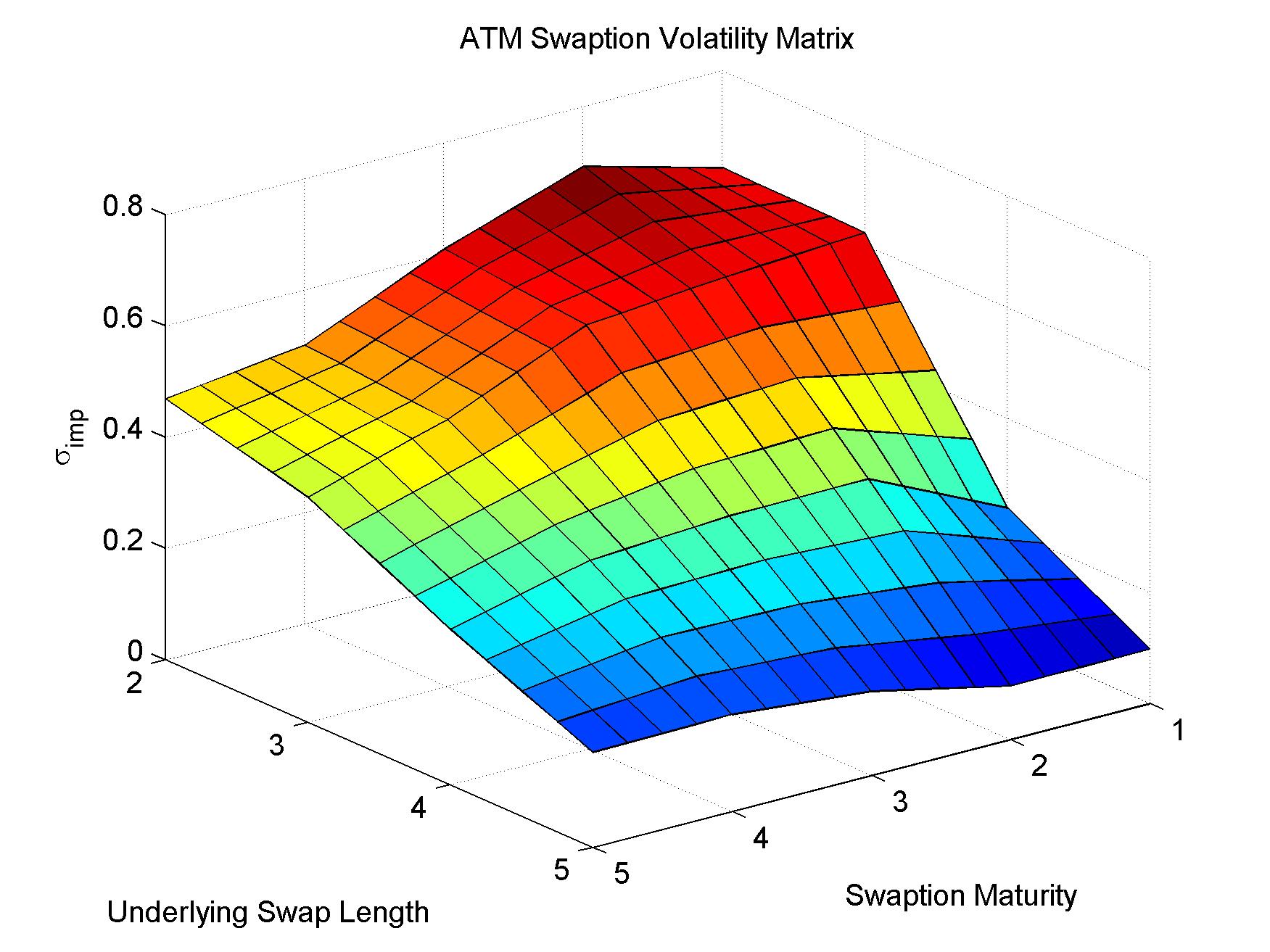}
  \caption{ATM Swaption Implied Volatility Surface generated by the Wishart Libor model}
  \label{fig:8}
\end{figure}

\begin{figure}[H]
\centering
  \includegraphics[scale=0.18]{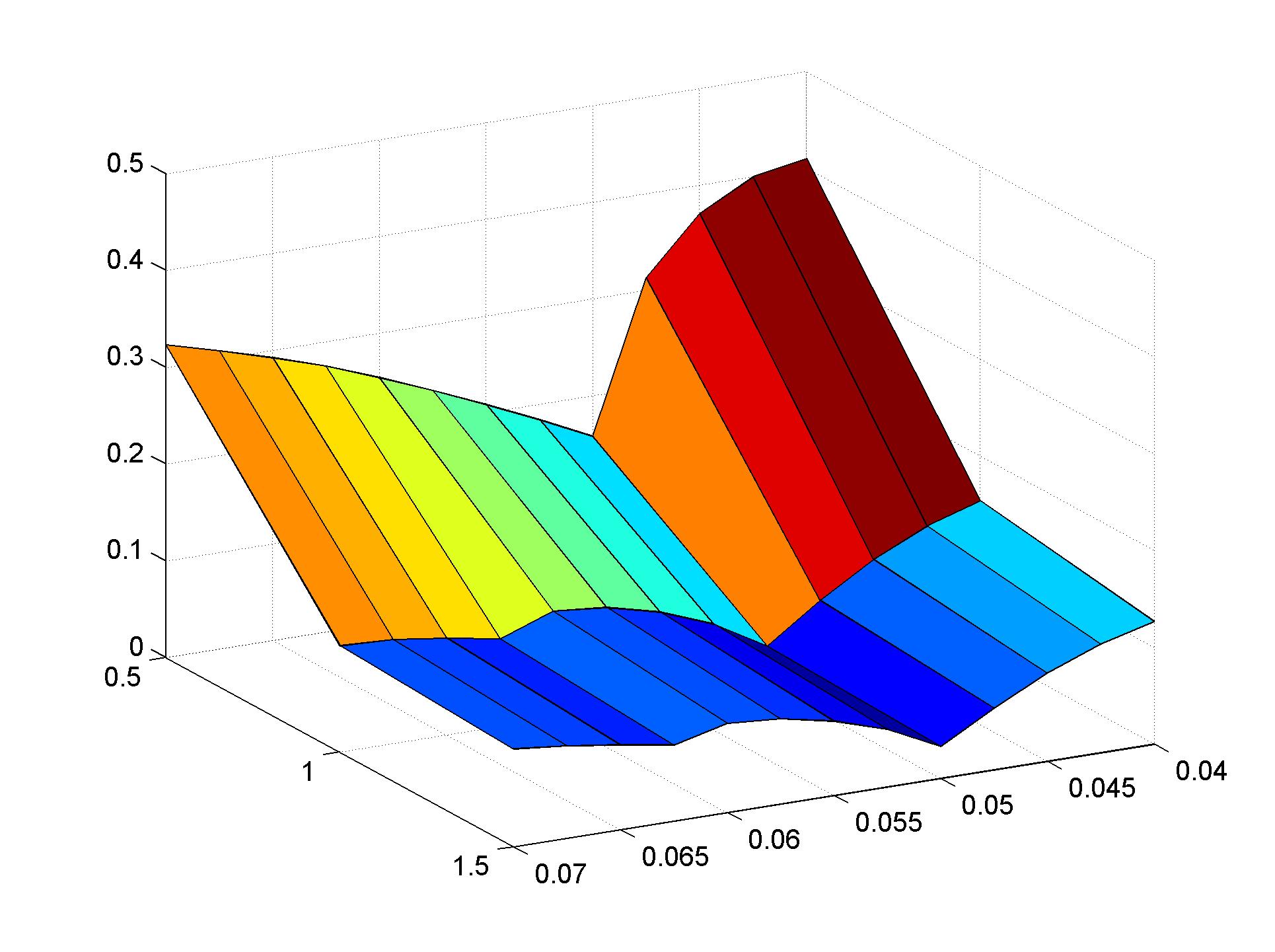}
  \caption{Caplet Implied Volatility Surface generated by the compound Poisson Libor model with Wishart distributed jumps}
  \label{fig:10}
\end{figure}

\clearpage

\bibliographystyle{abbrvnat}
\bibliography{biblio}
	
\addcontentsline{toc}{section}{Bibliography}

\end{document}